\documentclass[12pt]{amsart}
\usepackage{amsmath, amsfonts, amssymb}
\usepackage{color}

\newcommand{\memo}[1]{}

\DeclareMathOperator{\supp}{\rm supp}

\newcommand{\R}{{\mathbf R}}

\newcommand{\ve}{\varepsilon}
\newcommand{\pa}{\partial}

\newcommand{\jb}[1]{\left\langle #1 \right\rangle}

\newcommand{\norm}[2]{\|#1 \!:\! #2\|}
\newcommand{\D}{\Xi}
\newcommand{\Do}{\Omega}
\newcommand{\jbx}{\jb{x}}

\newcommand{\lam}{\lambda}
\newcommand{\vp}{\varphi}
\newcommand{\non}{\nonumber}
\newtheorem{theorem}{Theorem}[section]
\newtheorem{proposition}[theorem]{Proposition}

\newtheorem{lemma}[theorem]{Lemma}

\renewcommand{\thelemma}{\thesection.\arabic{lemma}}

\numberwithin{equation}{section}

\title[Nonlinear wave equations in an exterior domain in 2D]
{Almost global existence for nonlinear wave equations
in an exterior domain 
\\
in two space dimensions}

\author[H. Kubo]{Hideo Kubo}

\date{}

\subjclass{35L70, 35L20}

\begin{document}
\maketitle

\begin{center}
{\small
Graduate School of Information Sciences,
\\
Tohoku University, Sendai 980-8579, Japan
\\
kubo@math.is.tohoku.ac.jp}
\end{center}

\begin{abstract}
In this paper we deal with the exterior problem for a system of nonlinear wave equations 
in two space dimensions, assuming that the initial data is small and smooth.
We establish the same type of lower bound of the lifespan for the problem
as that for the Cauchy problem, despite of the weak decay property of the solution 
in two space dimensions.
\end{abstract}

\hspace{7mm}
{\it Keywords}. Exterior problem, Nonlinear wave equation, Lifespan

\begin{center}
{Dedicated to Professor Yoshihiro Shibata on the occasion of his 60th birthday}
\end{center}


\section{Introduction and statement of main results}
Let $\Do$ be an unbounded domain in $\R^n$~($n \ge 2$)
with compact and smooth boundary $\partial \Do$.
We put ${\mathcal O}:=\R^n \setminus {\Do}$, which is called
an obstacle and is assumed to be non-empty.
We consider the mixed problem for a system of nonlinear wave equations\,:
\begin{align}\label{ap1}
& (\partial_t^2-\Delta) u_i =F_i(\partial u, \nabla_x \pa u),
 & (t,x) \in (0,\infty)\times \Omega,
\\ \label{ap2}
& u(t,x)=0, & (t,x) \in (0,\infty)\times
\partial\Omega,
\\ \label{ap3}
& u(0,x)=\varepsilon \phi(x),\hspace{3mm}
\partial_t u(0,x)=\varepsilon \psi(x), & x\in \Omega
\end{align}
for $i=1, \dots, N$, where $u=(u_1,u_2\dots,u_N)$ is an unknown function,
$\Delta=\sum_{j=1}^n \partial_j^2$, $\partial_t=\partial_0=\partial/\partial t$,
$\partial_j=\partial/\partial x_j$ ($j=1,\ldots,n$), and $\ve>0$.
We assume $\phi$, $\psi \in C^\infty_0(\overline{\Do}\,;\R^N)$,
namely, they are smooth functions on $\overline{\Do}$
vanishing outside some ball. 
We also assume that $F_i(\partial u, \nabla_x \partial u)$ is a smooth function satisfying
\begin{align}\label{vo}
F_i(\partial u, \nabla_x \partial u)=O(|\partial u|^q+|\nabla_x \partial u|^q),
\quad 1 \le i \le N
\end{align}
around $(\partial u,\nabla_x \partial u)=0$ for some integer $q \ge 2$,
together with the energy symmetric condition.

We suppose, in addition, that $(\phi, \psi, F)$ satisfies the 
{\it compatibility condition} to infinite order
for the mixed problem \eqref{ap1}-\eqref{ap3}, that is,  
$(\partial^j_t u)(0,x)$, formally determined by \eqref{ap1} and \eqref{ap3},
vanishes on $\partial\Do$ for any non-negative integer $j$
$($notice that the values $(\partial^j_t u)(0,x)$ are determined by
$(\phi,\psi, F)$ successively {\rm;} for example we have
$\pa_t^2 u(0)=\Delta_x\phi+F(\psi, \nabla_x\phi)$, and so on$)$.

It was firstly shown by Shibata and Tsutsumi \cite{ShiTsu86} that
the mixed problem for (\ref{ap1})-(\ref{ap3}) admits a unique global solution 
for sufficiently small initial data, when either $n \ge 6$ and $q \ge 2$ or 
$3 \le n \le 5$ and $q \ge 3$,
provided ${\mathcal O}$ is non-trapping.
Although the dispersive property is getting weaker as the spatial dimension is
lower, there are already many contributions for the case where
$3 \le n \le 5$ and $q=2$
(see \cite{God89, God95}, \cite{Ha95}, \cite{KatKub}, \cite{KeSmiSo02G, KeSmiSo04}, \cite{Kub06}, \cite{Met04}, \cite{MetNaSo05b}, \cite{MetSo05}
and the references cited therein).  

However, up to the author's knowledge, there is no literature about the
exterior problem (\ref{ap1})-(\ref{ap3}) for the case $n=2$.
The aim of this paper is to treat the problem in that case,
by assuming that $q=3$ in \eqref{vo} and ${\mathcal O}$ is star-shaped.
We remark that when $n=2$, the cubic nonlinearity is on the critical level concerning
the global existence theorem for small initial data.
Indeed, if $N=1$ and $F_1=(\pa_t u)^3$, then 
one can show a blow-up result from a corresponding result for the Cauchy problem
(see e.g. \cite{Kub94}), because of the domain of dependance.

Let us denote the lifespan by $T_\ve$, i.e., the
supremum of all $T>0$ such that a classical solution to the problem
(\ref{ap1})-(\ref{ap3}) exists in $[0,T) \times \overline{\Omega}$. 
Then we find that $T_\ve \le \exp(A \ve ^{-2})$ holds for some positive constant 
$A$, in view of the argument in \cite{Kub94}.
Therefore, it is natural to ask whether the above upper bound of the lifespan 
is optimal with respect to $\varepsilon$ or not.
In this paper we shall establish an affirmative answer to this question as follows.

\begin{theorem}\label{thm:GE}
Let $n=2$ and let $\phi$, $\psi \in C^\infty(\overline{\Do}\,;\R^N)$
vanish outside certain ball.
Assume that $(\phi,\psi, F)$ satisfies the compatibility condition
to infinite order for the problem \eqref{ap1}-\eqref{ap3},
${\mathcal O}$ is star-shaped,
and $F$ satisfies \eqref{vo} with $q=3$.
Then there exist positive constants $\varepsilon_0$, $C$
such that for all $\ve \in (0,\ve_0]$,
we have $T_\ve \ge \exp(C\ve^{-2})$.
\end{theorem}

Our proof of the theorem is based on the cut-off method
used in \cite{ShiTsu86}.
Because the decaying rate of the local energy is actually weak when $n=2$
(see \eqref{obstacle} below), we need a careful treatment
for getting weighted pointwise estimates given in 
Theorem \ref{main} below from those for the corresponding 
Cauchy problem due to Kubota \cite{kubota}, Di Flaviano \cite{DiF03}
and Hoshiga and Kubo \cite{hk2}.
Unfortunately, the resulting pointwise estimate for derivatives of the solution
is not good enough, unlike the case of $n=3$,
for handling the boundary term arising from the integration-by-parts argument.
The main idea to overcome the difficulty is to make use of the stronger
decay property for the time derivative of the solution
than that for the space derivatives of the solution.
This stronger decay property for the time derivative 
is deduced from the fact that the boundary condition is preserved 
under the differentiation in time.

As in the work of  \cite{KatKub}, we shall use a
part of the vector fields of the Lorentz invariance: $\partial_t$, $\partial_j$~$(j=1,2)$, and 
$O_{12}=x_1 \partial_2-x_2\partial_1$,
because the boundary condition makes difficult to use
$t \partial_j+x_j \partial_t$ ($j=1, 2$) and $t\partial_t+x \cdot \nabla_x$.
We also remark that the geometric assumption on the obstacle will be used 
for assuring the decay of the local energy.

This paper is organized as follows. In the next section we collect several
notation.
In the section 3 we give some preliminaries needed later on.
The section 4 is devoted to establish the weighted pointwise estimates
(\ref{ba4}) and \eqref{ba4t} below.
Making use of these estimates, we give a proof of
the almost global existence theorem in the section 5.

\section{Notation}

For $\D=(v_0, v_1, f)\in H^1(\Omega)\times L^2(\Omega)\times 
L^\infty\bigl( (0,T); L^2(\Omega) \bigr)$, we denote by $S[\D](t,x)$ the solution
of the mixed problem\,:
\begin{align}\label{eq}
& (\partial_t^2-\Delta_x) v =f, & (t,x) \in (0,T)\times \Do,
\\ \label{dc}
& v(t,x)=0, & (t,x) \in (0,T)\times \partial\Do,
\\ \label{id}
& v(0,x)=v_0(x),\ (\partial_t v)(0,x)=v_1(x), & x\in \Do.
\end{align}
We sometimes write $\vec{v}_0=(v_0, v_1)$ in what follows.
We denote by $X(T)$ the set of all 
$$
\D=(v_0, v_1, f)=(\vec{v}_0,f)\in C^\infty_0(\overline{\Do}; \R^2)\times
C^\infty_D([0,T)\times \overline{\Do};\R)
$$
satisfying the compatibility condition to infinite order for \eqref{eq}-\eqref{id},
i.e., $(\pa_t^j v)(0,x)$, determined formally from \eqref{eq} and \eqref{id} by
\begin{align}\label{data+}
v_j(x)=\Delta v_{j-2}(x)+(\partial_t^{j-2} f)(0,x) \quad (j\ge 2)
\end{align}
vanishes on $\pa \Do$ for any non-negative integer $j$.
Here $f\in C^\infty_D([0,T)\times \overline{\Do};\R)$
means that $f\in C^\infty([0,T)\times \overline{\Do};\R)$
and $f(t,\cdot)\in C^\infty_0(\overline{\Do})$ for any fixed $t\in [0,T)$.
In addition, for $a>1$, $X_{a}(T)$ denotes 
the set of all $\D=(v_0, v_1, f)\in X(T)$ satisfying
$$ 
 v_0(x)=v_1(x)=f(t,x)\equiv 0\text{ for $|x|\ge a$ and $t\in [0, T)$.}
$$ 
Besides we set $K[\vec{v}_0]=S[(\vec{v}_0,0)]$ and
$L[f]=S[(0,0,f)]$.

Similarly, for 
$(w_0, w_1, g)\in H^1(\R^2)\times L^2(\R^2) \times 
    L^\infty\bigl((0,T); L^2(\R^2)\bigr)$,
we denote by $S_0[(w_0, w_1, g)](t,x)$ 
the solution of the following Cauchy problem\,:
\begin{align}\label{eq0}
&(\partial_t^2-\Delta_x) w = g, & (t,x) \in (0,T)\times \R^2,
\\ \label{id0}
& w(0,x)=w_0(x),\ (\partial_t w)(0,x)=w_1(x), & x\in \R^2.
\end{align}
Besides we put $K_0[\vec{w}_0]=S_0[(\vec{w}_0,0)]$ and
$L_0[g]=S_0[(0,0,g)]$, where $\vec{w}_0=(w_0, w_1)$.

We denote
\begin{equation}\nonumber 
Z_0=\partial_0=\partial_t, \quad Z_j=\partial_j \ (j=1,2), \quad
Z_3=O_{12}=x_1 \partial_2-x_2\partial_1.
\end{equation}
Then we have
\begin{align}\label{commute}
Z_j(\partial_t^2-\Delta)=(\partial_t^2-\Delta)Z_j
\quad \mbox{for}\  j=0, 1, 2, 3.
\end{align}
Denoting $Z^\alpha=Z_0^{\alpha_0} \cdots Z_{3}^{\alpha_{3}}$ with a multi-index
$\alpha=(\alpha_0, \ldots, \alpha_{3})$, we set
\begin{equation}\label{norm}
|\varphi (t,x)|_m=\sum_{|\alpha| \le m} |Z^\alpha \varphi(t,x)|,
\quad
\|\varphi(t)\|_m=\|\,|\varphi(t,\cdot)|_m\!:\!{L^2(\Do)}\|
\end{equation}
for a real or $\R^N$-valued smooth function $\varphi(t,x)$ and 
a non-negative integer $m$.

For $\nu$, $\kappa \in \R$, we define 
\begin{eqnarray}\label{defPhi}
&& {\Phi}_\nu(t,x)=
   \left\{
      \begin{array}{ll}
        \langle t+|x|\rangle^{\nu} &\quad \mbox{if} \hspace{3mm}  \nu<0,
\\
       \log^{-1}\bigg(2+\displaystyle\frac{\langle t+|x|\rangle}{\langle t-|x|\rangle}\bigg)
         &\quad \mbox{if} \hspace{3mm} \nu=0,
\\
       \langle t-|x|\rangle^{\frac{1}{2}} \langle t-|x|\rangle^{-[\frac{1}{2}-\nu]_+} 
         &\quad \mbox{if} \hspace{3mm}  \nu>0,   
      \end{array}
   \right. 
\\ \nonumber
\\ \label{defPsi}
&& \Psi_\kappa(t)=
\left\{\begin{array}{ll}
\log(2+t)  &\quad \mbox{if} \hspace{3mm} \kappa=1,
\\
1  &\quad \mbox{if} \hspace{3mm}  \kappa\not=1,
\end{array}\right.
\end{eqnarray}
where we have denoted
$$
A^{[a]_+}=A^a \quad \mbox{if} \quad a>0~;\
A^{[a]_+}=1 \quad \mbox{if} \quad a<0~;
\
A^{[0]_+}=1+\log A
$$
for $A\ge 1$, and $\jb{s}=\sqrt{1+|s|^2}$ for $s \in \R^n$.
Besides, for $\rho$, $\kappa \in {\mathbf R}$ and $c \ge 0$,
we put
\begin{align}\label{defz}
& z_{\rho,\kappa\,;\,c}(t,x)=\langle t+|x|\rangle^\rho
\langle ct-|x| \rangle^\kappa,
\\ \label{defW}
& W_{\rho,\kappa}(t,x)= 
 \langle t+|x|\rangle^\rho \Bigl( \min\{\jb{|x|},\jb{t-|x|} \}\Bigr)^\kappa,                          
\\
& w_\rho(t,x)=\jb{x}^{-1/2} \jb{t-|x|}^{-\rho}+\jb{t+|x|}^{-1/2} \jb{t-|x|}^{-1/2}.
\end{align}
Note that for $1/2 \le \rho \le 1$ we have
\begin{align} \label{el1}
& w_\rho(t,x) \le C \left( W_{1/2,1/2}(t,x) \right)^{-1},
\\ \label{el2}
& w_\rho(t,x) \le C \jb{t}^{-\rho} \quad {\rm if}\ |x| \le 2.
\end{align}

We define
\begin{equation}\label{NfW}
 \|f(t)\!:\!{N_k({\mathcal W})}\|
=\sup_{(s,x) \in [0,t] \times \Do} 
     \jbx^{1/2}\,{\mathcal W}(s,x)\,|f(s,x)|_k
\end{equation}
for $t\in [0,T)$, a non-negative integer $k$ and any non-negative
function $\mathcal{W}(s,x)$.
Similarly we put
\begin{equation}\label{eq:3.5}
 \|g(t)\!:\!{M_k(\mathcal{W})}\| =
   \sup_{(s,x) \in [0,t] \times \R^2} 
     \jbx^{1/2}\,\mathcal{W}(s,x)\,|g(s,x)|_k.
\end{equation}

Let $\rho\ge 0$, and $k$ be a non-negative integer. We define 
\begin{equation}
{\mathcal A}_{\rho, k}[v_0, v_1]=\sup_{x\in \Do} \jbx^{\rho} 
\bigl(|v_0(x)|_{k}+|\nabla_x v_0(x)|_k+|v_1(x)|_k\bigr)
\label{HomWei}
\end{equation}
for a smooth function $(v_0, v_1)$ on $\Do$, and
\begin{equation}
{\mathcal B}_{\rho, k}[w_0, w_1]=\sup_{x\in \R^2} \jbx^{\rho} 
\bigl(|w_0(x)|_{k}+|\nabla_x w_0(x)|_k+|w_1(x)|_k\bigr)
\label{HomWei0}
\end{equation}
for a smooth function $(w_0, w_1)$ on $\R^2$. 

For $a \ge 1$, let $\psi_a$ be a smooth radially symmetric function
on $\R^2$ satisfying
\begin{equation}\label{cutoff}
\psi_a(x)=0 \ (|x| \le a), \quad
\psi_a(x)=1 \ (|x| \ge a+1).
\end{equation}

We put $B_R=\{x\in \R^2;\, |x|< R\}$ for $R>0$.
We may assume, without loss of generality, that 
${\mathcal O}\subset B_{1}$ by the translation and scaling.
Hence we always assume it in the following.
For $R\ge 1$, we set $ \Do_R=\Do\cap {B_R}$.

\section{Preliminaries}

First we introduce an elliptic estimate,
whose proof will be given in Appendix A for
the sake of completeness.

\begin{lemma}\label{elliptic}\
Assume that ${\mathcal O}$ is star-shaped.
For $\varphi \in H^m(\Do) \cap H_0^1(\Do)$ with an integer $m\,(\ge 2)$, 
we have
\begin{equation}\label{ap10}
\sum_{|\alpha|=m} \norm{\partial_x^\alpha \varphi}{L^2(\Do)} \le
C(\|\Delta_x \varphi\!:\!{H^{m-2}(\Do)}\| + \|\nabla_x \varphi\!:\!{L^2(\Do)}\|).
\end{equation}
\end{lemma}

Next we derive an estimate for the local energy of solutions to \eqref{eq}-\eqref{id}. 
We put ${\mathcal H}^{m}(\Do)=H^{m+1}(\Do)\times H^{m}(\Do)$.

\begin{lemma}\label{LocalEnergy}
Assume that ${\mathcal O}$ is star-shaped.
Let $a$, $b>1$, $\gamma\in (0, 1]$, and $m$ be a non-negative
integer.
Then for $(\vec{v}_0, f) \in X_{a}(T)$, there exists a positive constant
$C=C(\gamma,a,b,m)$ such that for $t\in [0,T)$, we have
\begin{align}\label{LE}
& \sum_{|\alpha| \le m} \jb{t}^\gamma \norm{\pa^\alpha S[(\vec{v}_0, f)](t)}{L^2(\Do_b)}
\\ \nonumber
&\quad \le C \Bigl( \norm{\vec{v}_0}{{\mathcal H}^{m-1}(\Do)}
 {}+\sum_{|\alpha| \le m-1} \sup_{0\le s \le t} \jb{s}^{\gamma}  
  \norm{\pa^\alpha f(s)}{L^2(\Do)}
\Bigr).
\end{align}
\end{lemma}

\noindent{\it Proof.}\ 
It is known that there exists a
positive constant $C$ depending on $a$, $b$ such that
\begin{equation}
 \sum_{|\alpha|\le 1} \norm{\pa^\alpha K[\vec{\phi}_0](t)}{L^2(\Do_b)}
\le C \jb{t}^{-1} ({\rm log}(2+t))^{-2} \norm{\vec{\phi}_0}{{\mathcal H}^0(\Do)}
\label{obstacle}
\end{equation}
for any $\vec{\phi}_0=(\phi_0, \phi_1)\in H^1_0(\Do)\times L^2(\Do)$
satisfying $\phi_0(x)=\phi_1(x)\equiv 0$ for $|x|\ge a$
(see for instance Morawetz \cite{Mor75}, Vainberg \cite{Vai75}).

Now let $(\vec{v}_0, f)=(v_0, v_1, f)\in X_{a}(T)$.
Then, by Duhamel's principle, it follows that 
\begin{align}\label{DP} 
& \partial_t^j S[(\vec{v}_0,f)](t,x)
\\
& \qquad =K[(v_j,v_{j+1})](t,x)+
\int_0^t K[(0,(\partial_t^j f)(s))](t-s,x) ds \nonumber
\end{align}
for any non-negative integer $j$ and any $(t,x) \in [0,T) \times \Do$, 
where $v_j$ are given by (\ref{data+}).
Apparently we have $(\pa_t^j f)(s,\cdot)\in L^2(\Do)$
for $0\le s\le t$.
Thanks to the compatibility condition, we also find $v_j\in H^1_0(\Do)$
for any $j\ge 0$.
Therefore, by (\ref{obstacle}), for ${|\alpha| \le 1}$ and $\gamma \in (0,1]$, we have
\begin{align}\label{obstacle1}
& \norm{\partial^\alpha K[\vec{v}_j](t)}{L^2({\Do_b})}
 \le  \jb{t}^{-1} ({\rm log}(2+t))^{-2}  \norm{\vec{v}_j}{{\mathcal H}^{0}(\Do)}
\\
& \quad \le C  \jb{t}^{-\gamma} \bigl(\norm{\vec{v}_0}{{\mathcal H}^{j}(\Do)}
{}+\sum_{|\alpha| \le j-1} \norm{(\partial^\alpha f)(0)}{L^2(\Do)}\bigr)
\nonumber
\end{align}
with $\vec{v}_j=(v_j, v_{j+1})$, and
\begin{align}\label{obstacle2}
& \int_0^t \|\partial^\alpha K[(0,(\partial_t^j f)(s))](t-s) :L^2({\Do_b})\| ds
\\ \nonumber
&\qquad \le C \int_0^t  \jb{t-s}^{-1} ({\rm log}(2+t-s))^{-2} 
\,\norm{(\partial_t^j f)(s)}{L^2(\Do)} ds
\\ \nonumber
&\qquad \le C \jb{t}^{-\gamma} \sup_{0\le s \le t} \jb{s}^\gamma 
  \norm{(\partial_t^j f)(s)}{L^2(\Do)}.
\end{align}
Hence for ${|\alpha| \le 1}$ and $j \ge 0$, 
we get from \eqref{DP}
\begin{align}\label{LE1}
& 
 \| \partial^\alpha \partial^j_{t} S[(\vec{v}_0, f)](t)\!:\!{L^2(\Do_b)}\|
\\ \nonumber
& \quad \le C \jb{t}^{-\gamma} \bigl( \norm{\vec{v}_0}{{\mathcal H}^{j}(\Do)}
 +\sum_{|\alpha| \le j} \sup_{0\le s \le t} \jb{s}^\gamma 
\norm{\partial^\alpha f(s)}{L^2(\Do)}).
\end{align}

In order to evaluate $\pa_t^j \partial_x^\alpha v$ 
with ${|\alpha| \ge 2}$ and ${j+|\alpha| \le m}$, 
we make use of the following variant of (\ref{elliptic})\,:
\begin{equation}\label{LE2}
 \|\varphi\!:\!{H^m(\Do_b)}\| \le
C(\norm{\Delta_x \varphi}{H^{m-2}(\Do_{b^\prime})}
{}+\norm{\varphi}{H^1(\Do_{b^\prime})},
\end{equation}
where $1<b<b^\prime$ and $\varphi \in H^m(\Do) \cap 
H_0^1(\Do)$ with $m \ge 2$, together with the fact that
$\Delta_x S[(\vec{v}_0, f)]=\pa^2_t S[(\vec{v}_0, f)]+f$.
In this way, we obtain \eqref{LE}. 
This completes the proof.
\hfill$\qed$

\vspace{2mm}

Next we prepare three lemmas concerning the Cauchy problem. 
The first one is the decay estimate for
solutions of the homogeneous wave equation, due to \cite[Proposition 2.1]{kubota}
(observe that the general case can be reduced to the case $m=0$,
 thanks to (\ref{commute})).
We recall that ${\Phi}_\nu(t,x)$ and ${\Psi}_\kappa(t)$ were defined by
\eqref{defPhi} and \eqref{defPsi}, respectively.

\begin{lemma}\label{lem:freeH}\
For $\vec{v}_0\in (C_0^\infty({\mathbf
R}^2))^2$, $\nu>0$ and a non-negative integer $m$, 
there is a positive constant $C=C(\nu,m)$ such that
\begin{align}\label{decay}
\langle t+|x| \rangle^{1/2}\,{\Phi}_{\nu-1}(t,x)
   |K_{0}[\vec{v}_0](t,x)|_m
  \le C {\mathcal B}_{\nu+(1/2), m}[\vec{v}_0]
\end{align}
for $(t,x) \in [0,T) \times {\mathbf R}^2$.
\end{lemma}

The second one is the decay estimates for the inhomogeneous wave equation.

\begin{lemma}\label{free}\
If $\nu>0$, $\kappa \ge 1$ and $m$ is a non-negative integer, then
there exists a positive constant $C=C(\nu,\kappa,m)$ such that
\memo{
\begin{eqnarray}\label{po}
&& \langle t+|x| \rangle^{1/2}\,{\Phi}_{\nu-1}(t,x) |L_{0}[g](t,x)|_m
 \le C \Psi_\kappa(t+|x|)\,\|g(t)\!:\!{M_m(z_{\nu,\kappa\,;\,c})}\|
\\ \nonumber
&& \hspace{65mm} 
  +C \sum_{|\alpha| \le m-1}
      \|\,\langle \cdot \rangle^{\nu+(1/2)} (\Gamma^\alpha g)(0)\!:\!L^\infty({\mathbf R}^2)\|
\end{eqnarray}
and
}
\begin{align}\label{ba1}
 & \langle t+|x| \rangle^{1/2}\,{\Phi}_{\nu-1}(t,x) |L_{0}[g](t,x)|_m
\\ \notag  & \quad 
\le C \Psi_\kappa(t+|x|)\,\|g(t)\!:\!{M_m(W_{\nu,\kappa})}\|
\end{align}
for $(t,x) \in [0,T) \times {\mathbf R}^2$.
\end{lemma}

\noindent{\it Proof.}\ 
Let $|\alpha| \le m$.
It follows from (\ref{commute}) that
\begin{eqnarray}\label{comI}
Z^\alpha L_0[g]=L_0[Z^\alpha g]+K_0[(\phi_\alpha,\psi_\alpha)],
\end{eqnarray}
where we put
$$
\phi_\alpha(x)=(Z^\alpha L_0[g])(0,x), \quad
\psi_\alpha(x)=(\partial_t Z^\alpha L_0[g])(0,x).
$$
From the equation (\ref{eq0}) we get
$$
\phi_\alpha(x)= \sum_{|\beta| \le |\alpha|-2} C_\beta
(Z^\beta g)(0,x), 
\quad \psi_\alpha(x)= \sum_{|\beta| \le
|\alpha|-1} C_\beta^\prime (Z^\beta g)(0,x)
$$
with suitable constants $C_\beta$, $C_\beta^\prime$. 
Therefore, by (\ref{decay}), we get
\begin{align*}
&   \langle t+|x| \rangle^{1/2}\,{\Phi}_{\nu-1}(t,x) |K_0[(\phi_\alpha,\psi_\alpha)]| 
\\
& \quad  \le  C \sum_{|\alpha| \le m-1}
   \|\,\langle \cdot \rangle^{\nu+(1/2)} (Z^\alpha g)(0)\!:\!L^\infty({\mathbf R}^2)\|
\\
& \quad \le C \|g(0)\!:\!{M_{m-1}(W_{\nu,\kappa})}\|.
\end{align*}
Hence, in view of \eqref{comI}, it is enough to show \eqref{ba1} for $m=0$.

If we set $z_{\nu,\kappa;c}(s,y)=\jb{s+|y|}^{\nu} \jb{|y|-cs}^{\kappa}$, then we have
$$
(W_{\nu,\kappa}(s,y))^{-1} \le (z_{\nu,\kappa,0}(s,y))^{-1}+(z_{\nu,\kappa;1}(s,y))^{-1},
$$
so that
\begin{align}\label{ba1-}
& |L_{0}[g](t,x)| 
\le \|g(t)\!:\!{M_0(W_{\nu,\kappa})}\|\, |L_{0}[W_{\nu,\kappa}^{-1}](t,x)| 
\\ \notag
& \quad \le 
\|g(t)\!:\!{M_0(W_{\nu,\kappa})}\|
 \,( |L_{0}[z_{\nu,\kappa;0}^{-1}](t,x)| + |L_{0}[z_{\nu,\kappa;1}^{-1}](t,x)| )
\end{align}
for $(t,x) \in [0,T) \times {\mathbf R}^2$.
Since it was shown by \cite[Proposition 3.1]{DiF03} that
\begin{align}\label{DiF}
\langle t+|x| \rangle^{1/2}\,{\Phi}_{\nu-1}(t,x) |L_{0}[z_{\nu,\kappa;1}^{-1}](t,x)| 
 \le C \Psi_\kappa(t)
\end{align}
holds, we have only to show 
\begin{align}\label{DiF0}
\langle t+|x| \rangle^{1/2}\,{\Phi}_{\nu-1}(t,x) |L_{0}[z_{\nu,\kappa;0}^{-1}](t,x)| 
 \le C \Psi_\kappa(t).
\end{align}
%

Following the proof of \cite[Proposition 3.1]{DiF03}, we obtain
\begin{align}\nonumber
  |L_{0}[z_{\nu,\kappa;0}^{-1}](t,x)|
 & \le C \left\{
      \int_{|t-r|}^{t+r} \langle \alpha \rangle^{-\nu+(1/2)}
      (\alpha+r-t)^{-1/2}\,V(\alpha) d\alpha
      \right.
\\ \nonumber
& \qquad  \left. 
      +\phi(r,t) \int_{0}^{t-r} \langle \alpha \rangle^{-\nu+(1/2)}
      (t-r-\alpha)^{-1/{2}}\,V(\alpha) d\alpha
      \right\},
\end{align}
where we put $r=|x|$,
\begin{eqnarray}\nonumber
&&     V(\alpha)=\int_{-\alpha}^{\alpha} \langle \frac{\alpha+\beta}2 \rangle^{-\kappa+(1/2)}
                    (\beta+r+t)^{-1/2}d\beta,
\\ \nonumber
&&       \phi(r,t)=
      \left\{
          \begin{array}{ll}
	  0 & \ \text{if} \ \ 0\leq t\leq r,\\
	  1 & \ \text{if} \ \ t> r
	  \end{array}
      \right. 
\end{eqnarray}
(Notice that $\alpha+r-t>0$ if $\alpha>|t-r|$).
Therefore, once we find
\begin{eqnarray}\label{v1}
V(\alpha) \le C \langle \alpha \rangle^{1/2}\,\langle \alpha 
\rangle^{[1-\kappa]_+}\,\langle t+r \rangle^{-1/2}
\end{eqnarray}
for $0<\alpha<t+r$, then we get \eqref{DiF0} by Lemmas 3.3 and 3.4
in \cite{DiF03}.

We are going to show \eqref{v1}.
When $0<t+r<1$ or $(t+r)/2 <\alpha <t+r$,
it suffices to show that $V(\alpha) \le C \langle \alpha \rangle^{[1-\kappa]_+}$.
Splitting the integral at $\beta=-\alpha+1$, we get
\begin{eqnarray}\nonumber
 V(\alpha)\le \int_{-\alpha}^{-\alpha+1} (\beta+\alpha)^{-1/2}d\beta
          +C\int_{-\alpha+1}^{\alpha} \langle \frac{\alpha+\beta}2 \rangle^{-\kappa}
             d\beta,
\end{eqnarray}
since $\beta+r+t>\beta+\alpha>(\beta+\alpha+1)/2$ if $\alpha<t+r$ and $\beta>-\alpha+1$. 
This estimate yields (\ref{v1}).
On the other hand, when $t+r>1$ and $0<\alpha<(t+r)/2$,
we have $\beta+r+t>(t+r)/2$ if $\beta>-\alpha$, hence
\begin{eqnarray}\nonumber
 V(\alpha)\le C(t+r)^{-1/2}
         \int_{-\alpha}^{\alpha} \langle \frac{\alpha+\beta}2 \rangle^{-\kappa+(1/2)}
             d\beta.
\end{eqnarray}
The last integral is bounded by $C\langle \alpha \rangle^{[-\kappa+(3/2)]_+}
\le C\langle \alpha \rangle^{1/2}$ when $\kappa\ge 1$.
Thus we obtain (\ref{v1}). 
This completes the proof.
\hfill$\qed$

\vspace{2mm}

The third one is the decay estimates for derivatives of solutions of the
inhomogeneous wave equation.

\begin{lemma}\label{freeD}\
If $0<\nu<3/2$, $\mu\ge 0$, $\kappa \ge 1$, $0<\eta<1$,
and $m$ is a non-negative integer, then
there exists a positive constant $C=C(\nu,\kappa,\mu,m)$ such that we have
\begin{align} \label{der1m}
& (w_\nu(t,x))^{-1} |\pa_{t,x} L_0[g](t,x)|_m
\\ \notag
& \quad \le C \Psi_{1+\mu}(t+|x|) 
  \Psi_{\kappa}(t+|x|)\,\|g(t)\!:\!{M_{m+1}(z_{\nu+\mu,\kappa;0})}\|,
\end{align}
\begin{align} \label{ba2}
& (w_{1-\eta}(t,x))^{-1}  |\partial_{t,x} L_{0}[g](t,x)|_m
\\ \notag
& \quad \le C \log(2+t+|x|) \,\|g(t)\!:\!{M_{m+1}(W_{1,1})}\|
\end{align}
for $(t,x) \in [0,T) \times {\mathbf R}^2$.
\end{lemma}

\noindent{\it Proof.}\ 
Let $|\alpha| \le m$.      
Then we have
\begin{eqnarray}\label{comII}
\partial_{t,x} Z^\alpha L_0[g]=\partial_{t,x} L_0[Z^\alpha g]
      +\partial_{t,x} K_0[(\phi_\alpha,\psi_\alpha)]
\end{eqnarray}
from (\ref{comI}), and 
\begin{eqnarray}\nonumber
&& \partial_\ell L_{0}[Z^\alpha g]
 =L_{0}[\partial_\ell Z^\alpha g] \quad (\ell=1,2),
\\ \label{TimeDer}
&& \partial_t L_{0}[Z^\alpha g]
 =L_{0}[\partial_t Z^\alpha g]
  +K_0[(0,(Z^\alpha g)(0))].
\end{eqnarray}

First we prove \eqref{der1m}. By \eqref{der1} it suffices to show 
\begin{align} \label{free1}
& (w_\nu(t,x))^{-1}  |\partial_{t,x} K_{0}[(\phi_\alpha,\psi_\alpha)](t,x)|
\\ \notag
& \quad \le C \Psi_{1+\mu}(t+|x|) 
 \,\|g(0)\!:\!{M_{m+1}(z_{\nu+\mu,1;0})}\|,
\\ \label{free2}
& (w_\nu(t,x))^{-1}    |K_0[(0,(Z^\alpha g)(0))](t,x)|
\\ \notag
& \quad \le C \Psi_{1+\mu}(t+|x|) 
  \,\|g(0)\!:\!{M_{m+1}(z_{\nu+\mu,1;0})}\|.
\end{align}
We shall show only \eqref{free1}, because the proof of the other is similar.

When $0<\nu<1/2$, by (\ref{decay}) with $\nu$
replaced by $1+\nu$, we get
\begin{align}\nonumber
& \langle t+|x| \rangle^{1/2}\langle t-|x| \rangle^{\nu}
   |\partial_{t,x} K_{0}[(\phi_\alpha,\psi_\alpha)](t,x)|
\\ \notag
& \quad \le C \sum_{|\alpha| \le m}
      \|\,\langle \cdot \rangle^{(3/2)+\nu} (Z^\alpha g)(0)\!:\!L^\infty({\mathbf R}^2)\|
\\ \notag
& \quad \le C \| g(0): M_m(z_{\nu,1;0})\|.
\end{align}
On the other hand, when $\nu\ge 1/2$, by (\ref{decay}) with $\nu
=(3/2)+\mu$, we get
\begin{align}\label{free3}
& \langle t+|x| \rangle^{1/2}\langle t-|x| \rangle^{1/2}
   |\partial_{t,x} K_{0}[(\phi_\alpha,\psi_\alpha)](t,x)|
\\ \notag
& \quad \le C\Psi_{1+\mu}(t-|x|) \| g(0): M_m(z_{(1/2)+\mu,1;0})\|.
\end{align}
Thus we obtain \eqref{free1}.

Next we prove \eqref{ba2}. It follows from \eqref{der1} and \eqref{der1Bis}
with $\nu=1-\eta$, $\mu=\eta$ and $\kappa=1$ that
\begin{align*} 
 (w_{1-\eta}(t,x))^{-1} |L_{0}[\partial_{t,x}g](t,x)|_m
 \le C \log(2+t+|x|) \,\|g(t)\!:\!{M_{m+1}(W_{1,1})}\|.
\end{align*}
Since it is easy to see that \eqref{free3} is still valid
if we replace $z_{(1/2)+\mu,1;0}$ by $z_{(1/2)+\mu,1;1}$, we find \eqref{ba2}.
This completes the proof.
\hfill$\qed$

\vspace{2mm}

Finally, we introduce the following Sobolev type inequality, whose
counterpart for the Cauchy problem is due to Klainerman \cite{kl0}.
Since the proof is similar to the case of $n=3$ (see e.g. \cite{KatKub}),
we omit it.

\begin{lemma}\label{KlainermanSobolev}\
Let $\varphi \in C_0^2(\overline{\Do})$.
Then we have
\begin{equation}\label{ap21}
\sup_{x \in \Do} \jbx^{1/2} |\varphi(x)|
\le C \sum_{|\alpha|+k \le 2} \| \pa_x^\alpha O_{12}^k\, \varphi\!:\!{L^2(\Do)}\|.
\end{equation}
\end{lemma}

\section{Basic estimates}

First of all, we prepare the following lemma which will be used to 
prove Theorem \ref{main} below.
Recall that we have assumed ${\mathcal O}\subset B_1$.

\begin{lemma}\label{KataLem}
Let ${\mathcal O}$ be star-shaped.
Let $a$, $b>1$, $0<\rho\le 1$, $\mu\ge 0$, $\kappa\ge 1$, $\eta>0$,
and $m$ is a non-negative integer.

\noindent
{\rm (i)} Suppose that $\chi$ is a smooth function on $\R^2$ satisfying $\supp \chi \subset B_b$. 
If $\D=(\vec{v_0},f) \in X_{a}(T)$,  
then there exists a positive constant $C=C(\rho, a, b, m)$ such that
\begin{align} \label{KataL01}
 & \langle t \rangle^\rho
   |\chi S[ \D  ](t,x)|_m
\\ \nonumber
& \quad \le C  {\mathcal A}_{\rho+1, m+1}[\vec{v}_0]
 +C \sum_{|\beta|\le m+1} \sup_{(s,x)\in [0,t]\times\Do_a} 
 \langle s\rangle^{\rho} |\pa^\beta f(s,x)|
\end{align}
for $(t,x)\in[0, T)\times \overline{\Do}$.

\medskip

\noindent
{\rm (ii)} Let $\vec{w}$ and $g$ are smooth functions on $\R^2$ and on $[0, T)\times \R^2$,
respectively.
If  
$\supp g(t,\cdot)\subset \overline{B_a\setminus B_1}$
for any $t\in [0,T)$,
then there exists a positive constant $C=C(\rho, \mu, a, m)$ such that
\begin{align} 
\label{KataL02}
   |L_0[g](t,x)|_m
 \le C \sum_{|\beta|\le m} \sup_{(s,x)\in [0,t]\times\Do_a} 
 \langle s\rangle^{1/2} |\pa^\beta g(s,x)|
\end{align}
and
\begin{align} 
\label{KataL03}
  & (w_{\rho}(t,x))^{-1} |\pa L_0[g ](t,x)|_m
\\ \notag
& \quad 
 \le C\Psi_{1+\mu}(t+|x|)\!\! \sum_{|\beta|\le m+1} \sup_{(s,x)\in [0,t]\times\Do_a}
 \langle s\rangle^{\rho+\mu} |\pa^\beta g(s,x)|
\end{align}
for $(t,x)\in [0,T)\times \overline{\Do}$.

On the other hand, if $\vec{w_0}(x)=g(t,x)=0$ for any $(t, x) \in [0,T) \times B_1$,
then there exists a positive constant $C=C(\kappa,b,m)$ such that 
\begin{align} 
\label{KataL04}
 & \langle t \rangle^{1/2} |S_0[(\vec{w_0}, g) ](t,x)|_m
\\ \nonumber
& \quad \le C {\mathcal A}_{3/2, m}[\vec{w}_0]
 +C \Psi_\kappa(t+|x|)  \norm{g(t)}{N_m(W_{1,\kappa})}
\end{align}
and
\begin{align} 
\label{KataL05}
  & \langle t \rangle^{1-\eta} |\pa S_0[(\vec{w_0}, g) ](t,x)|_m
  \le C {\mathcal A}_{2, m+1}[\vec{w}_0]
\\ \notag
& \quad \quad
  +C\log(2+t+|x|) \norm{g(t)}{N_{m+1}(W_{1,1})}
\end{align}
for $(t,x)\in [0,T)
\times \overline{\Do_b}$.
\end{lemma}

\noindent{\it Proof.}\ \
First we note that 
\begin{equation}
|h(t,x)|_m\le C \sum_{|\beta|\le m} |\pa^\beta h(t,x)|
\label{KataM01}
\end{equation}
holds for any smooth function $h$ on $[0,T)\times \overline{\Do}$
(or on $[0, T)\times \R^2$) with $\supp h(t, \cdot)\subset B_R$
for some $R\,(>1)$.

We start with the proof of \eqref{KataL01}.
Let $\D=(\vec{v}_0, f) \in X_{a}(T)$ and $0<\rho\le 1$.
For $(t,x)\in [0, T)\times \overline{\Omega}$, 
by \eqref{KataM01}, 
the Sobolev inequality and \eqref{LE}, we get
\begin{align*}
\langle t \rangle^\rho
&|\chi S[ \D  ](t,x)|_m
\\
  \le & C \langle t\rangle^\rho \!\!\! \sum_{|\beta|\le m+2}
        \norm{\pa^\beta S[ \D ](t)}{L^2(\Do_b)}
\\
\le & C \norm{ \vec{v}_0}{{\mathcal H}^{m+1}(\Do)}
 +C \sup_{s\in [0,t]} \langle s\rangle^{\rho} \!\!\! 
\sum_{|\beta|\le m+1} \norm{\pa^\beta f(s)}{L^2(\Do)}.
\end{align*}
Since $\text{supp}\,f(t,\cdot) \subset \overline{{\Do_{a}}}$
implies $\norm{\pa^\beta f(s)}{L^2(\Do)}\le C \norm{\pa^\beta f(s)}{L^\infty(\Do_a)}$,
we obtain \eqref{KataL01}.


Next we prove \eqref{KataL02}.
By \eqref{ba1} with $\nu=1/2$ and $\kappa>1$, we find that the left-hand side
on \eqref{KataL02} is estimated by
\begin{align*}
 C \norm{g(t)}{M_m(W_{1/2,\kappa})} 
 \le C \norm{g(t)}{N_m(W_{1/2,\kappa})}
\end{align*}
because $\supp g(t,\cdot) \subset \overline{B_a\setminus B_1}$.
Since $|x| \le a$ on $\supp g(t,x)$, we get \eqref{KataL02}.
Similarly, if we use \eqref{der1m} with $\nu=\rho$ and $\kappa>1$,
instead of \eqref{ba1}, then we get \eqref{KataL03}. 

Next we prove \eqref{KataL04}.
From \eqref{decay} and \eqref{ba1} we have
\begin{align*}
& \langle t+|x| \rangle^{1/2} \Phi_0(t,x)
  |S_0[(\vec{w}_0, g) ](t,x)|_m
\\
& \quad \le C {\mathcal B}_{3/2, m}[\vec{w}_0]
 +C \Psi_\kappa(t+|x|)\,\norm{g(t)}{M_m(W_{1,\kappa})}
\end{align*}
for $(t,x)\in [0, T)\times \R^2$.
Since $\Phi_0(t,x)$ is equivalent to a constant
when $x \in \overline{\Omega_b}$,
we get \eqref{KataL04}, because of the assumption on $\vec{w}_0$ and $g$.

Finally, we prove \eqref{KataL05}.
From 
\eqref{decay} and \eqref{ba2} we have
\begin{align*}
& (w_{1-\eta}(t,x))^{-1}  |\pa S_0[(\vec{w}_0, g) ](t,x)|_m
\\
& \quad \le C {\mathcal B}_{2, m+1}[\vec{w}_0]
 +C \log(2+t+r)\,\norm{g(t)}{M_{m+1}(W_{1,1})}
\end{align*}
for $(t,x)\in [0, T)\times \R^2$.
Recalling \eqref{el2},
we get \eqref{KataL04}.
This completes the proof.
\hfill$\qed$

\vspace{2mm}

Now we are in a position to state our basic estimates for solutions
to the linear mixed problem.

\begin{theorem}\ \label{main}
Let ${\mathcal O}$ be star-shaped and $\D=(\vec{v}_0,f) \in X(T)$.
Let $k$ be a nonnegative integer and $\kappa \ge 1$.

\noindent 
{\rm (i)}\ It holds that
\begin{align}\label{ba3}
    |S[\D](t,x)|_k
\le C( {\mathcal A}_{3/2,k+3}[\vec{v_0}]
  +\Psi_\kappa(t+|x|) \| f(t)\!:\!{N_{k+3}(W_{1,\kappa})}\|)
\end{align}
for $(t,x)\in [0,T)\times \overline{\Do}$.

\vspace{1mm}

\noindent
{\rm (ii)}\ For $0<\rho \le 1/2$ and $\delta>0$,  we have
\begin{align}\label{ba4}
&   |\partial S[\D](t,x)|_k
       \le C w_{1/2}(t,x)  \left({\mathcal A}_{2+\delta,k+4}[\vec{v_0}]  \right.
\\ \notag
& \quad\quad  \ \left.
    +\log (2+t+|x|)\,\|f(t)\!:\!{N_{k+3}(W_{1,1})}\| \right)
\\ \notag
& \quad\quad 
  +C w_{\rho}(t,x) \Psi_{(3/2)-\rho}(t+|x|)\, \Psi_\kappa(t+|x|)\,\|f(t)\!:\!{N_{k+4}(W_{1,\kappa})}\| 
\end{align}
for $(t,x)\in [0,T)\times \overline{\Do}$.

Moreover, for $0< \eta <1$ and $\delta>0$, we have
\begin{align}\label{ba4t}
& |\pa \partial_t S[\D](t,x)|_k
     \le C w_{1-\eta}(t,x) \left({\mathcal A}_{2+\delta,k+5}[\vec{v_0}] \right. 
\\ \notag
& \quad \quad \left.
 +({\rm log} (2+t+|x|))^2\,\|f(t)\!:\!{N_{k+5}(W_{1,1})}\| \right)
\end{align}
for $(t,x)\in [0,T)\times \overline{\Do}$.
\end{theorem}

\noindent{\it Proof.}\ \
First we prove \eqref{ba3}.
We use the following representation formula
based on the cut-off method\,:
\begin{equation}\label{v31}
S[{\D}](t,x)=\psi_1(x) S_0[\psi_2 \D](t,x)
{}+\sum_{i=1}^4 S_i[\D](t,x)
\end{equation}
for $(t,x)\in [0,T)\times \overline{\Do}$,
where $\psi_a$ is defined by (\ref{cutoff}) and we have set 
\begin{align}\label{S1}
& S_1[\D](t,x)=(1-\psi_2(x))L\bigl[\,[\psi_1,-\Delta_x]
              S_0[\psi_2 \D ]\bigr](t,x),
\\ \label{S2}
& S_2[\D](t,x) 
=-L_0\bigl[\,[\psi_2,-\Delta_x]
  L\bigl[\,[\psi_1,-\Delta_x]
             S_0[\psi_2 \D]\bigr]\bigr](t,x),
\\ \label{S3}
& S_3[\D](t,x)=(1-\psi_3(x)) S[(1-\psi_2) \D](t,x),
\\ \label{S4}
& S_4[\D](t,x)
=-L_0\bigl[\,[\psi_3,-\Delta_x] S[(1-\psi_2)\D]\bigr](t,x).
\end{align}

By \eqref{decay} and \eqref{ba1} we have
\begin{align*}
& \jb{t+|x|}^{1/2} \Phi_0(t,x) \left| \psi_1(x) S_0[\psi_2\D](t,x) \right|_k 
\\
& \quad\le C {\mathcal A}_{3/2,k}[\vec{v}_0]
 + C \Psi_\kappa(t+|x|) \|f(t):N_{k}(W_{1,\kappa})\|.
\end{align*}

Next we shall estimate $S_1[\D]$ and $S_3[\D]$.
It is easy to check that
$$
[\psi_a,-\Delta_x]h(t,x)=
   h(t,x) \Delta_x \psi_a(x)+2\nabla_{\!x}\,h(t,x) \cdot \nabla_{\!x}\, \psi_a(x)
$$
for $(t, x) \in [0,T)\times \overline{\Do}$, $a \ge 1$ and any smooth function $h$.
Note that this identity implies $(0,0, [\psi_a, -\Delta_x]h)\in X_{a+1}(T)$
for any smooth function $h$ and $a\ge 1$, because 
$\supp \nabla_x \psi_a\cup \supp \Delta_x \psi_a\subset 
\overline{B_{a+1}\setminus B_a}$.
Therefore, by \eqref{KataL01} with $\rho=1/2$ and \eqref{KataL04}, we obtain
\begin{align}\label{poD1}
& 
\langle t \rangle^{1/2} |S_1[\D](t,x)|_k
\\ \nonumber
 &\quad \le C \sum_{|\beta|\le k+2} 
\sup_{(s,x)\in [0,t]\times\Do_2} 
 \langle s\rangle^{1/2} |\pa^\beta S_0[\psi_2 \D](s,x)|
\\ \nonumber
 & \quad \le C {\mathcal A}_{3/2,k+2}[\vec{v_0}]
   +C \Psi_\kappa(t+|x|) \,\|f(t)\!:\!{N_{k+2}(W_{1,\kappa})}\|
   \nonumber
\end{align}
for $(t,x)\in [0,T)\times \overline{\Do}$.
Similarly, since we have $(1-\psi_2)\D\in X_{3}(T)$
for any $\D\in X(T)$, \eqref{KataL01} with $\rho=1/2$ leads to 
\begin{align}\label{poD3}
&  \langle t \rangle^{1/2} |S_3[\D](t,x)|_k
\le C {\mathcal A}_{3/2,k+1}[\vec{v_0}]
\\
& \qquad 
+C \sum_{|\beta|\le k+1} \sup_{(s,x)\in [0,t]\times\Do_3} 
 \langle s\rangle^{1/2} |\pa^\beta f(s,x)|
\nonumber
\end{align}
for $(t,x)\in [0,T)\times \overline{\Do}$.
Since ${\rm supp}\, S_i[\Xi] \subset B_4$~$(i=1, 3)$, 
\eqref{poD1} and \eqref{poD3} imply 
\begin{align} \notag
&  \langle t+|x| \rangle^{1/2} |S_i[\D](t,x)|_k
\le C {\mathcal A}_{3/2,k+2}[\vec{v_0}]
\\
& \quad \hspace{10mm} 
+C \Psi_\kappa(t+|x|)  \|f(t):N_{k+2}(W_{1,\kappa})\|
\quad (i=1,3)
\nonumber
\end{align}
for $(t,x)\in [0,T)\times \overline{\Do}$.

Set $g_i[\D]=(\pa_t^2-\Delta_x) S_i[\D]$ for $i=2, 4$.
Observing that $g_2$ and $g_4$ have the almost same structures as $S_1$ and $S_3$, respectively, we find
\begin{align} \notag
&  \sum_{|\beta|\le m} \sup_{(s,x)\in [0,t]\times\Do_4} 
 \langle s\rangle^{1/2} |\pa^\beta g_i[\D] (s,x)|
\le C {\mathcal A}_{3/2,m+2}[\vec{v_0}]
\\
& \quad \hspace{15mm} 
+C \Psi_\kappa(t+|x|)  \|f(t):N_{m+2}(W_{1,\kappa})\|
\nonumber
\end{align}
for $i=2, 4$ and $(t,x)\in [0,T)\times \overline{\Do}$, in a similar fashion.
Note that $g_2$ and $g_4$ are supported on 
$\overline{B_4\setminus B_2}$.
By \eqref{S2}, \eqref{S4}, and \eqref{KataL02}, we obtain
\begin{align} \notag
 |S_i[\D] (t,x)|_k
\le C {\mathcal A}_{3/2,k+3}[\vec{v_0}]
+C \Psi_\kappa(t+|x|)  \|f(t):N_{k+3}(W_{1,\kappa})\|
\end{align}
for $i=2, 4$ and $(t,x)\in [0,T)\times \overline{\Do}$,
and hence \eqref{ba3}.

Next we prove \eqref{ba4} by using \eqref{v31}.
Let $0<\rho \le 1/2$ and $\delta>0$.
Writing $\zeta_0=S_0[\psi_2\D]$, we get
\begin{align*}
& \left|\pa_a \bigl(\psi_1(x) \zeta_0(t,x)\bigr)\right|_k 
\le C|\pa_a \zeta_0(t,x)|_k
{}+C|\pa_a \psi_1(x)|_k |\zeta_0(t,x)|_k.
\end{align*}
By \eqref{decay} and \eqref{ba2} with $\eta=1/2$, we see that the first term on the right-hand
side is estimated by 
\begin{align*}
& C \jb{t+|x|}^{-1/2}\jb{t-|x|}^{-1/2} {\mathcal A}_{2+\delta,k+1}[\vec{v}_0]
\\ 
& \quad +C w_{1/2}(t,x)  \log (2+t+|x|)\,\|f(t):N_{k+1}(W_{1,1})\|.
\end{align*}
Since $\jb{t}^{-1/2} \le C \jb{x}^{-1/2}\jb{t-|x|}^{-1/2}$ for
$|x| \le 2$, \eqref{KataL04} shows that the second term on the right-hand
side is estimated by 
\begin{align*}
 C \jb{x}^{-1/2}\jb{t-|x|}^{-1/2} ( {\mathcal A}_{3/2,k}[\vec{v}_0]
 +\log (2+t+|x|)  \|f(t):N_{k}(W_{1,1})\|).
\end{align*}
Therefore we have
\begin{align}\label{pi0}
& |\pa (\psi_1(x) S_0[\psi_2\D](t,x))|_k 
 \le C w_{1/2}(t,x) ( {\mathcal A}_{2+\delta,k+1}[\vec{v}_0]
\\ \notag
& \quad +\log (2+t+|x|)\,\|f(t):N_{k+1}(W_{1,1})\|)
\end{align}
for $(t,x)\in [0,T)\times \overline{\Do}$.

From \eqref{poD1} with $\kappa=1$ and \eqref{poD3} we have
\begin{align} \label{pi1}
&  \jb{x}^{1/2} \langle t-|x| \rangle^{1/2} |\pa S_i[\D](t,x)|_k
\le C {\mathcal A}_{3/2,k+3}[\vec{v_0}]
\\
& \quad \hspace{10mm} 
+C \log (2+t+|x|)  \|f(t):N_{k+3}(W_{1,1})\|
\quad (i=1,3)
\nonumber
\end{align}
for $(t,x)\in [0,T)\times \overline{\Do}$.

As for $g_4[\D]=(\pa_t^2-\Delta_x) S_4[\D]$, it follows from
\eqref{KataL01} with $\rho=1$ that
\begin{align} \notag
&  \sum_{|\beta|\le m} \sup_{(s,x)\in [0,t]\times\Do_4} 
 \langle s\rangle |\pa^\beta g_4[\D] (s,x)|
\\
& \quad \le C {\mathcal A}_{2,m+2}[\vec{v}_0]
+C \|f(t):N_{m+2}(W_{1,1})\|.
\nonumber
\end{align}
Therefore, by \eqref{S4} and \eqref{KataL03} with $\rho=\mu=1/2$,
we get
\begin{align} \label{pi2}
&  (w_{1/2}(t,x))^{-1} |\pa S_4[\D](t,x)|_k
\\ \notag
&  \quad \le C {\mathcal A}_{2,k+3}[\vec{v}_0]
+C \|f(t):N_{k+3}(W_{1,1})\|
\end{align}
for $(t,x)\in [0,T)\times \overline{\Do}$.

To estimate $g_2[\D]= 
-[\psi_2,-\Delta_x]  L\bigl[\,[\psi_1,-\Delta_x]
             S_0[\psi_2 \D]\bigr]$,
we define $g_{2,0}[\vec{v}_0]$ and $g_{2,1}[f]$ by replacing $S_0[\psi_2 \D]$ with
$K_0[\psi_2 \vec{v}_0]$ and $L_0[\psi_2 f]$, respectively.
Then we have $g_2[\D]=g_{2,0}[\vec{v}_0]+g_{2,1}[f]$.
By \eqref{KataL01} with $\rho=(1/2)+\mu$ and \eqref{decay} with
$\nu=1+\mu$~$(0<\mu<1/2)$, we get
\begin{align} \notag
&  \jb{t}^{(1/2)+\mu} |g_{2,0}[\vec{v}_0](t,x)|_m
\\ \notag
& \quad  \le C \sum_{|\beta| \le m+3} \sup_{(s,x) \in [0,t] \times \Omega_2}
    \jb{s}^{(1/2)+\mu} |\pa^{\beta} K_0[\psi_2 \vec{v}_0](s,x)|
\\ \notag
& \quad \le C {\mathcal A}_{(3/2)+\mu,m+3}[\vec{v_0}]
\end{align}
for $(t,x)\in [0,T)\times \overline{\Do}$.
Applying \eqref{KataL03} with $\rho=1/2$ and $\mu>0$, we find
\begin{align} \notag
  (w_{1/2}(t,x))^{-1} |\pa L_0[g_{2,0}[\vec{v}_0]](t,x)|_k
\le C {\mathcal A}_{(3/2)+\mu,k+4}[\vec{v_0}]
\end{align}
for $(t,x)\in [0,T)\times \overline{\Do}$.

On the other hand, similarly to the proof of \eqref{poD1}, we get
\begin{align*} 
\langle t \rangle^{1/2} |g_{2,1}[f](t,x)|_m
 \le C \Psi_\kappa(t+|x|)\,\|f(t)\!:\!{N_{m+3}(W_{1,\kappa})}\|
   \nonumber
\end{align*}
for $(t,x)\in [0,T)\times \overline{\Do}$.
By \eqref{KataL03} with $\mu=(1/2)-\rho$, we obtain
\begin{align} \notag
& (w_{\rho}(t,x))^{-1} |\pa L_0[g_{2,1}[f]](t,x)|_k
\\ \notag
& \quad \le C 
\Psi_{(3/2)-\rho}(t+|x|)\, \Psi_\kappa(t+|x|) \,\|f(t)\!:\!{N_{k+4}(W_{1,\kappa})}\|
\end{align}
for $(t,x)\in [0,T)\times \overline{\Do}$. 
Thus we get
\begin{align} \label{pi3} 
&  |\pa S_2[\Xi](t,x)|_k
\le C w_{1/2}(t,x) {\mathcal A}_{(3/2)+\mu,k+4}[\vec{v_0}]
\\ \notag
& \quad +C w_{\rho}(t,x)
\Psi_{(3/2)-\rho}(t+|x|)\, \Psi_\kappa(t+|x|)\,\|f(t)\!:\!{N_{k+4}(W_{1,\kappa})}\|
\end{align}
for $(t,x)\in [0,T)\times \overline{\Do}$. 
Now \eqref{ba4} follows from \eqref{pi0}, \eqref{pi1}, \eqref{pi2}, and \eqref{pi3}.

Finally we prove \eqref{ba4t}.
Let $0<\eta<1$ and $\delta>0$.
Since $\pa_t S[{\D}]$ satisfies the boundary condition, we have
\begin{equation}\label{v31t}
\pa_t S[{\D}](t,x)=\psi_1(x) S_0[\psi_2 \widetilde{\D}](t,x)
{}+\sum_{i=1}^4 S_i[\widetilde{\D}](t,x)
\end{equation}
for $(t,x)\in [0,T)\times \overline{\Do}$,
where we have set $\widetilde{\D}=(v_1,v_2,\pa_t f)$ with
$v_2=\Delta v_0+f$. 
Since $S_0[\psi_2 \widetilde{\D}]=\pa_t S_0[\psi_2 {\D}]$, we find 
from  \eqref{decay} and \eqref{ba2} that
\begin{align} \label{pit}
& |\pa (\psi_1(x) S_0[\psi_2 \widetilde{\D}](t,x)) |_k
\\ \notag
& \quad   \le C\jb{t+|x|}^{-1/2}\jb{t-|x|}^{-1/2} {\mathcal A}_{2+\delta, k+2}[\vec{v}_0]
\\ \notag
& \quad \ \ 
 +C w_{1/2}(t,x) \log(2+t+|x|)\,\|f(t):N_{k+2}(W_{1,1})\|
\end{align}
for $(t,x)\in [0,T)\times \overline{\Do}$. 

By \eqref{KataL01} and \eqref{KataL05}, we obtain
\begin{align}\label{pi0t}
& 
\langle t \rangle^{1-\eta} |\pa S_1[\widetilde{\D}](t,x)|_k
\\ \nonumber
 &\quad \le C \sum_{|\beta|\le k+3} 
\sup_{(s,x)\in [0,t]\times\Do_2} 
 \langle s\rangle^{1-\eta} |\pa^\beta S_0[\psi_2 \widetilde{\D}](s,x)|
\\ \nonumber
 & \quad \le C {\mathcal A}_{2,k+4}[\vec{v}_0]
 +C \log(2+t+|x|)\, \|f(t)\!:\!{N_{k+4}(W_{1,1})}\|
\end{align}
for $(t,x)\in [0,T)\times \overline{\Do}$.

Since $S[(1-\psi_2)\widetilde{\Xi}]=\pa_t S[(1-\psi_2){\Xi}]$,
by \eqref{KataL01} we have
\begin{align}\label{pi01t}
&  \langle t \rangle |\pa S_3[\widetilde{\D}](t,x)|_k
\le C {\mathcal A}_{2,k+3}[\vec{v}_0]
\\
& \qquad 
+C \sum_{|\beta|\le k+3} \sup_{(s,x)\in [0,t]\times\Do_3} 
 \langle s\rangle |\pa^\beta f(s,x)|
\nonumber
\end{align}
for $(t,x)\in [0,T)\times \overline{\Do}$.

Next we evaluate $S_4[\widetilde{\Xi}]$. 
Since $S_4[\widetilde{\Xi}]=L_0[\pa_t g_4[\Xi]]$,
analogously to the proof of \eqref{pi2}, we get 
\begin{align} \label{pi02t}
&  (w_{1-\eta}(t,x))^{-1} |\pa S_4[\widetilde{\Xi}](t,x)|_k
\\ \notag
&  \quad \le C {\mathcal A}_{2,k+4}[\vec{v}_0]
+C \|f(t):N_{k+4}(W_{1,1})\|
\end{align}
for $(t,x)\in [0,T)\times \overline{\Do}$, if we use \eqref{KataL03}
with $\rho=1-\eta$, $\mu=\eta$.

Next we evaluate $S_2[\widetilde{\Xi}]$.
We define $\widetilde{g}_{2,0}[\vec{v}_0]$ and $\widetilde{g}_{2,1}[f]$ 
by replacing $\pa_t S_0[\psi_2 \D]$ with
$\pa_t K_0[\psi_2 \vec{v}_0]$ and $\pa_t L_0[\psi_2 f]$
in 
$$
g_2[\widetilde{\Xi}]= -[\psi_2,-\Delta_x]  L\bigl[\,[\psi_1,-\Delta_x]
             \pa_t S_0[\psi_2 \D]\bigr],
$$
respectively, so that $g_2[\widetilde{\Xi}]=\widetilde{g}_{2,0}[\vec{v}_0]+
\widetilde{g}_{2,1}[f]$.
By \eqref{KataL01} and \eqref{decay}, we get
\begin{align} \notag
 \jb{t}^{1-\eta} |\widetilde{g}_{2,0}[\vec{v}_0](t,x)|_m
&  \le C \sum_{|\beta| \le m+3} \sup_{(s,x) \in [0,t] \times \Omega_2}
    \jb{s}^{1-\eta} |\pa^{\beta}\pa_t K_0[\psi_2 \vec{v}_0](s,x)|
\\ \notag
& \le C {\mathcal A}_{2,m+4}[\vec{v}_0]
\end{align}
for $(t,x)\in [0,T)\times \overline{\Do}$.
Appling \eqref{KataL03} with $\rho=1-\eta$, $\mu=\eta$,
we find
\begin{align} \notag
  (w_{1-\eta}(t,x))^{-1} |\pa L_0[\widetilde{g}_{2,0}[\vec{v}_0]](t,x)|_k
\le C {\mathcal A}_{2,k+5}[\vec{v}_0]
\end{align}
for $(t,x)\in [0,T)\times \overline{\Do}$.
On the other hand, by \eqref{KataL01} and \eqref{KataL05}, we obtain
\begin{align} \notag
\langle t \rangle^{1-\eta} |g_{2,1}[f](t,x)|_m
 & \le C \sum_{|\beta|\le m+3} 
\sup_{(s,x)\in [0,t]\times\Do_2} 
 \langle s\rangle^{1-\eta} |\pa^\beta \pa_t L_0[\psi_2 f](s,x)|
\\ \nonumber
 & \le C \log(2+t+|x|)\,\|f(t)\!:\!{N_{m+4}(W_{1,1})}\|
   \nonumber
\end{align}
for $(t,x)\in [0,T)\times \overline{\Do}$.
By \eqref{KataL03} with $\rho=1-\eta$ and $\mu=0$, we obtain
\begin{align} \notag
& (w_{1-\eta}(t,x))^{-1} |\pa L_0[g_{2,1}[f]](t,x)|_k
\\ \notag
& \quad \le C (\log(2+t+|x|))^2\,\|f(t)\!:\!{N_{k+5}(W_{1,1})}\|
\end{align}
for $(t,x)\in [0,T)\times \overline{\Do}$. Thus we get
\begin{align} \label{pi03t} 
& (w_{1-\eta}(t,x))^{-1} |\pa S_2[\widetilde{\Xi}](t,x)|_k
\le C {\mathcal A}_{2,k+5}[\vec{v_0}]
\\ \notag
& \quad \quad +C (\log(2+t+|x|))^2\,\|f(t)\!:\!{N_{k+5}(W_{1,1})}\|
\end{align}
for $(t,x)\in [0,T)\times \overline{\Do}$. 
Now \eqref{ba4t} follows from \eqref{pit}, \eqref{pi0t}, \eqref{pi01t}, \eqref{pi02t}, and \eqref{pi03t}.
This completes the proof. \hfill$\qed$

\section{Proof of Theorem 1.4}

In this section we prove Theorem \ref{thm:GE}.
Let all the assumptions of Theorem \ref{thm:GE} be fulfilled
and let us assume ${\mathcal O}\subset B_{1}$.
Though there is no essential difficulty in treating the quasilinear 
case, we concentrate on the semilinear case to keep our exposition simple.
Namely, we assume that the nonlinearity is of the form 
\begin{align}\label{ap4}
F_i(\partial u) =
\sum_{a=0}^2 \sum_{j,k,l=1}^N g^{a,b,c}_i (\partial_a u_j)(\partial_b u_k)(\partial_c u_l),
\end{align}
where $g^{a,b,c}_i$~$(a,b,c=0,1,2$) are real constants.

Since the local existence for the mixed problem \eqref{ap1}-\eqref{ap3}
can be proved by a standard argument, we have only to deduce 
an {\it apriori} estimate. 
Let $u$ be a smooth solution to \eqref{ap1}-\eqref{ap3} on $[0,T)\times \overline{\Do}$.
For a nonnegative integer $k$, we set
\begin{align*}
e_{k}[u](t,x)= (w_{1/2}(t,x))^{-1} |\pa u(t,x)|_{k}. 
\end{align*}
Since $\phi$, $\psi \in C^\infty_0(\overline{\Do}; \R^N)$,
we have $\norm{e_{k}[u](0)}{L^\infty(\Do)} \le C\ve$.
Let $k\ge 29$, and assume that
\begin{align}
\sup_{0\le t<T} \norm{e_{k}[u](t)}{L^\infty(\Do)} \le M\ve
\label{InductiveAs}
\end{align}
holds for some large $M(>1)$ and small $\ve(>0)$, satisfying $M\ve \le 1$.

Because the decay property is weak when $n=2$, we need to refine 
a treatment of the boundary term arising from the integration-by-parts argument,
compared with the case $n=3$.
Namely, we shall make use of rather stronger decay of the time derivative
based on \eqref{ba4t}.

\subsection{Estimates of the energy}\label{KEE1}
In this subsection we shall prove
\begin{equation}
\sum_{|\alpha|\le 2k} \norm{\pa^\alpha \pa u(t)}{L^2(\Do)}\le 
CM\ve (1+t)^{C_0M^2\ve^2},
\quad t \in [0,T)
\label{ap11}
\end{equation}
where $C_0$ is a universal constant which is independent of $M$, $\ve$ and $T$.

For $0\le m\le 2k$, we define
$z_{m}(t)=\sum_{p=0}^{2k-m} \norm{\pa_t^p \pa u(t)}{H^m(\Do)}$.
To prove \eqref{ap11}, it suffices to show
\begin{equation}
z_{m}(t) \le CM\ve (1+t)^{C_0M^2\ve^2} \quad\text{for $0\le m\le 2k$.}
\label{bird}
\end{equation}
First we evaluate $z_0(t)$.
For $0\le p\le 2k$, from \eqref{InductiveAs} we get 
\begin{equation}\nonumber
 |\pa_t^{p} F(\pa u)(t,x)|
 \le C M^2\ve^2 (1+t)^{-1} \sum_{q=0}^{2k} |\pa_t^{q} \pa u(t,x)|,
\end{equation}
so that 
$$
\|\pa_t^{p} F(\pa u)(t)\!:\!{L^2(\Do)}\| 
\le C_0 M^2\ve^2 (1+t)^{-1} z_{0}(t). 
$$
Therefore, noting that the boundary condition (\ref{ap2}) implies
$\partial_t^p u(t,x)=0$ for $(t,x)\in [0,T) \times \partial \Do$
and $0\le p\le 2k+1$, 
we see from the energy inequality 
for the wave equation that 
\begin{equation}\nonumber
\frac{d z_{0}}{dt}(t) \le C_0 M^2\ve^2 (1+t)^{-1} z_{0}(t),
\end{equation}
which yields
\begin{equation}\label{ene1}
z_{0}(t) \le CM\ve (1+t)^{C_0 M^2\ve^2}.
\end{equation}

Next suppose $m\ge 1$.
Then, from the definition of $z_m$, we have
\begin{align*}
z_m(t)\le & C\sum_{p=0}^{2k-m}\bigl(\norm{\pa_t^p \pa u(t)}{L^2(\Do)}
{}+\sum_{1\le |\alpha|\le m}
\norm{\pa_t^p \pa_x^\alpha \pa_t u(t)}{L^2(\Do)}\\
&\qquad\qquad\qquad{}+\sum_{1\le |\alpha|\le m}\norm{\pa_t^p \pa_x^\alpha \nabla_x u(t)}{L^2(\Do)}
\bigr)\\
\le & C\bigl(z_0(t)+z_{m-1}(t)+\sum_{p=0}^{2k-m}\sum_{2\le |\alpha|\le m+1}
\norm{\pa_t^p\pa_x^\alpha u(t)}{L^2(\Do)}\bigr),
\end{align*}
where we have used
$$
\sum_{1\le |\alpha|\le m}
\norm{\pa_t^p \pa_x^\alpha \pa_t u(t)}{L^2(\Do)}
\le C \sum_{|\alpha'|\le m-1}\norm{\pa_t^{p+1}\pa_x^{\alpha'}\nabla_x u(t)}{L^2(\Do)}.
$$
For $2\le |\alpha|\le m+1$, (\ref{ap10}) yields
\begin{equation}\nonumber
\|\partial_t^p \pa_x^\alpha u(t)\!:\!{L^2(\Do)}\|
\le C( \|\Delta_x \partial_t^{p} u(t)\!:\!{{H}^{m-1}(\Do) }\|
+\|\nabla_{\!x}\, \partial_t^{p} u(t)\!:\!{L^2(\Do)}\|).
\end{equation}
For $0\le p \le 2k-m$, we see 
that the second term on the right-hand side in the above is bounded by $z_0(t)$.
On the other hand, by using (\ref{ap1}), the first term is estimated by
\begin{align*}
& C( \|\partial_t^{p+2} u(t)\!:\!{{H}^{m-1}(\Do) }\|
+\|\partial_t^{p} F(\pa u)(t)\!:\!{H^{m-1}(\Do)}\|)
\\
& \quad \le C(z_{m-1}(t)+M^2\ve^2 (1+t)^{-1} z_{m-1}(t))
\end{align*}
for $0\le p\le 2k-m$.
Since $M\ve \le 1$, we obtain
\begin{equation}\label{nanana}
 z_m(t) \le C\bigl( z_{m-1}(t)+z_0(t) \bigr)
\end{equation}
for $m\ge 1$.
Using \eqref{ene1},
we find \eqref{ap11}.

\subsection{Estimates of the generalized energy, part 1}\label{KEE2}
In this subsection we evaluate the generalized derivatives $\pa Z^\alpha u$ in $L^2(\Do)$
for $1 \le |\alpha| \le 2k-1$.
It follows from (\ref{commute}) that
\begin{eqnarray}\label{ene2}
&& \quad \frac12\frac{d}{dt} 
 \int_{\Do} \left(|\partial_t Z^\alpha u_i|^2+|\nabla_{\!x}\, Z^\alpha u_i|^2
  \right)\,dx
\\ \nonumber
&& =\int_{\Do} Z^\alpha F_i(\partial u)\,\partial_t Z^\alpha u_i\,dx
 +\int_{\partial \Do} (\nu\cdot \nabla_{\!x}\, Z^\alpha u_i)\,(\partial_t
  Z^\alpha u_i)\,dS,
\end{eqnarray}
where $\nu=\nu(x)$ is the unit outer normal vector at $x \in \partial \Do$,
and $dS$ is the surface measure on $\partial \Do$.
Since $\partial \Do \subset B_{1}$ implies
$|\pa Z^\alpha u(t,x)| \le C\sum_{1 \le |\beta| \le |\alpha|} |\partial^\beta\pa u(t,x)|$
for $(t,x)\in [0,T) \times \partial \Do$,
we see from the trace theorem that the second term on the right-hand side
of (\ref{ene2}) is evaluated by
$$
C \sum_{1 \le |\beta| \le |\alpha|+1} \|\partial^\beta \partial u(t)\!:\!{L^2(\Do)}\|^2
 \le CM^2\ve^2 (1+t)^{2C_0M^2\ve^2},
$$
in view of (\ref{ap11}). On the other hand, it follows that
\begin{align}\label{ene3}
\|Z^{\alpha} F(\pa u)(t)\!:\!{L^2(\Do)}\| 
\le & C_1 M^2\ve^2 (1+t)^{-1} \|\pa u(t)\|_{|\alpha|}
\end{align}
for ${|\alpha| \le 2k-1}$, where $C_1$ is a constant independent 
of $\alpha$, $M$, $\ve$, and $T$.
 
Let $\delta>0$ be a sufficiently small number that is fixed later on, 
and take $\ve_0>0$ in such a way that $C_1 M^2 \ve_0^2 + C_0 M^2 \ve_0^2 \le \delta$.
Then for $0<\ve\le \ve_0$, we get
\begin{align*}
\frac{d}{dt}\|\pa u(t)\|_{2k-1}^2 \le 
C_1 M^2\ve^2 (1+t)^{-1} \|\pa u(t)\|_{2k-1}^2
+CM^2\ve^2 (1+t)^{2\delta},
\end{align*}
which leads to 
\begin{equation}\label{ap15}
\|\pa u(t)\|_{2k-1}\le CM\ve (1+t)^{\delta+(1/2)},
\quad t \in [0,T).
\end{equation}

\subsection{Estimates of the generalized energy, part 2}\label{KEE3}
By (\ref{ap21}) and (\ref{ap15}) we have
\begin{align}\label{ap21-}
\jbx^{1/2} |\partial u(t,x)|_{2k-3} 
 \le  CM\ve (1+t)^{\delta+(1/2)},
\end{align}
which yields
\begin{align*} 
\| F(\pa u)(t)\,:\,N_{2k-3}(W_{1,1})\|
 \le CM^3 \ve^3 (1+t)^{\delta+(1/2)},
\end{align*}
by \eqref{el1} with $\nu=1/2$.
For a sufficiently small number $\eta$, we set $\nu=1-\eta$, 
in the following.
Applying \eqref{ba4t}, we get
\begin{align*} 
 |\pa \partial_t u(t,x)|_{2k-8}
     \le C w_{\nu}(t,x) (\ve + M^3 \ve^3 ({\rm log} (2+t))^2 
    (1+t)^{\delta+(1/2)}),
\end{align*}
because of the finite speed of propagation.
When $(t,x)\in [0,T)\times \overline{\Do_1}$, by \eqref{el2} we have
\begin{align*} 
|\pa \partial_t u(t,x)|_{2k-8}
 \le CM\ve (1+t)^{-\nu+2\delta+(1/2)},
\end{align*}
because $M>1$ and $M \ve \le 1$.
Using this inequality, \eqref{ap11}, and \eqref{ene3}, we arrive at
\begin{align*}
\frac{d}{dt}\|\pa u(t)\|_{2k-7}^2 \le 
C_1 M^2\ve^2 (1+t)^{-1} \|\pa u(t)\|_{2k-7}^2
+CM^2\ve^2 (1+t)^{-\nu+3\delta+(1/2)},
\end{align*}
which leads to 
\begin{equation}\label{ap16}
\|\pa u(t)\|_{2k-7}\le CM\ve (1+t)^{-(\nu/2)+(3\delta/2)+(3/4)},
\quad t \in [0,T),
\end{equation}
provided $0<\ve\le \ve_0$.

\subsection{Estimates of the generalized energy, part 3}\label{KEE4}
Repeating the argument in the previous step, we get
\begin{align*} 
\| F(\pa u)(t)\,:\,N_{2k-9}(W_{1,1})\|
 \le CM^3 \ve^3  (1+t)^{-(\nu/2)+(3\delta/2)+(3/4)}.
\end{align*}
Using \eqref{ba4} with $\rho=1/2$, $\kappa=1$, and \eqref{ba4t} with $\nu=1-\eta$, we obtain
\begin{align*} 
 &|\pa \nabla u(t,x)|_{2k-14}
\\ \notag
& \quad
     \le C w_{1/2}(t,x) (\ve + M^3 \ve^3 ({\rm log} (2+t))^2)
     (1+t)^{-(\nu/2)+(3\delta/2)+(3/4)}
\\ \notag
& \quad \le CM\ve  (1+t)^{-(\nu/2)+2\delta+(1/4)},
\\
 &|\pa \partial_t u(t,x)|_{2k-14}
\\ \notag
& \quad
     \le C w_{\nu}(t,x) (\ve + M^3 \ve^3 ({\rm log} (2+t))^2 )
     (1+t)^{-(\nu/2)+(3\delta/2)+(3/4)}
\\ \notag
& \quad \le CM\ve  (1+t)^{-(3\nu/2)+2\delta+(3/4)}
\end{align*}
for $(t,x)\in [0,T)\times \overline{\Do_1}$.
Therefore, we get
\begin{align*}
\frac{d}{dt}\|\pa u(t)\|_{2k-13}^2 \le 
C_1 M^2\ve^2 (1+t)^{-1} \|\pa u(t)\|_{2k-13}^2
+CM^2\ve^2 (1+t)^{-2\nu+4\delta+1},
\end{align*}
which yields
\begin{equation}\label{ap17}
\|\pa u(t)\|_{2k-13}\le CM\ve (1+t)^{-\nu+2\delta+1},
\quad t \in [0,T),
\end{equation}
provided $0<\ve\le \ve_0$. 

\subsection{Estimates of the generalized energy, part 4}\label{KEE5}
As before, we have
\begin{align} 
\| F(\pa u)(t)\,:\,N_{2k-15}(W_{1,1})\|
 \le CM^3 \ve^3  (1+t)^{-\nu+2\delta+1},
\end{align}
so that
\begin{align*} 
 &|\pa \nabla u(t,x)|_{2k-20}
   \le CM\ve  (1+t)^{-\nu+3\delta+(1/2)},
\\ \notag
 &|\pa \partial_t u(t,x)|_{2k-20}
   \le CM\ve  (1+t)^{-2\nu+3\delta+1}
\end{align*}
for $(t,x)\in [0,T)\times \overline{\Do_1}$.
Hence we get
\begin{align*}
\frac{d}{dt}\|\pa u(t)\|_{2k-19}^2 \le 
C_1 M^2\ve^2 (1+t)^{-1} \|\pa u(t)\|_{2k-19}^2
+CM^2\ve^2 (1+t)^{-3\nu+6\delta+(3/2)}.
\end{align*}
If we choose $\delta$ so small that $3\nu-6\delta-(3/2) >1$, then we get 
$\|\pa u(t)\|_{2k-19}\le CM\ve (1+t)^{C_1M^2\ve^2}$ for $0<\ve\le \ve_0$.
Taking $T$ in such a way that
\begin{equation}\label{cond1}
(2+T)^{M^2\ve^2} \le e,
\end{equation}
we have
\begin{equation}\label{ap18}
\|\pa u(t)\|_{2k-19}\le CM\ve, \quad t \in [0,T).
\end{equation}

\subsection{Pointwise estimates, part 2}
By virtue of \eqref{ap18}, we get
\begin{align*} 
\| F(\pa u)(t)\,:\,N_{2k-21}(W_{1,1})\|  \le CM^3 \ve^3.
\end{align*}
Let $0<\rho <1/2$. Then it follows from \eqref{ba4} that
\begin{align*}
|\pa u(t,x)|_{2k-25}
     \le C w_{\rho}(t,x) (\ve + M^3 \ve^3 {\rm log} (2+t))
\end{align*}
for $(t,x)\in [0,T)\times \overline{\Do}$.
Therefore, assuming \eqref{cond1}, we get
\begin{align} \label{ap20}
|\pa u(t,x)|_{2k-25}
 \le CM\ve \,  w_{\rho}(t,x) 
\end{align}
for $(t,x)\in [0,T)\times \overline{\Do}$, provided $0<\ve\le \ve_0$. 

\subsection{Pointwise estimates, final part}
For $\kappa=1+\rho$, we get 
\begin{align*} 
\| F(\pa u)(t)\,:\,N_{2k-25}(W_{1,\kappa})\|  \le CM^3 \ve^3,
\end{align*}
by \eqref{ap20}.
Using \eqref{ba4} with $\rho=1/2$ and $\kappa>1$, we have
\begin{align} \label{KFF}
|\pa u(t,x)|_{2k-29}
     \le C_2 w_{1/2}(t,x) (\ve + M^3 \ve^3 {\rm log} (2+t) )
\end{align}
for $(t,x)\in [0,T)\times \overline{\Do}$, provided $0<\ve\le \ve_0$. 
Here $C_2$ is a constant independent of $M$, $\ve$ and $T$.
From \eqref{KFF} we find that \eqref{InductiveAs} with $M$ replaced by
$M/2$ is true for $M\ge 4C_2$ and $C_2 M^2\ve^2 \log (2+T) \le 1/4$,
Then, for $\ve \in (0,\ve_0]$,
the standard continuity argument implies that $e_k[u](t)$ stays bounded
as long as the solution $u$ exists (observe that $\norm{e_k[u](t)}{L^\infty(\Do)}$
is continuous with respect to $t$, because $u$ is smooth and
$\supp u(t, \cdot)\subset B_{t+R}$ for $t\in[0, T)$ with some $R>0$). Theorem \ref{thm:GE}
follows immediately from this {\it a priori} bound and a restriction on $T$.
This completes the proof.
\hfill$\qed$

\renewcommand{\theequation}{A.\arabic{equation}}
\setcounter{equation}{0}  
\renewcommand{\thelemma}{A.\arabic{theorem}}
\renewcommand{\thetheorem}{A.\arabic{theorem}}
\setcounter{theorem}{0}
\renewcommand{\thesubsection}{A.\arabic{subsection}}
\setcounter{subsection}{0}
\section*{Appendix A: Proof of Lemma \ref{elliptic}}
Suppose $m\ge 2$ and $\varphi\in H^m(\Omega)\cap H^1_0(\Omega)$.
Let $\chi$ be a $C^\infty_0(\R^2)$ function such that
$\chi \equiv 1$ in a neighborhood of ${\mathcal O}$.
Let $\text{supp}\,\chi \subset B_R$ for some $R>1$.
We set $\varphi_1=\chi \varphi$ and $\varphi_2=(1-\chi) \varphi$, so that $\varphi=\varphi_1+\varphi_2$.

First we estimate $\varphi_1$.
The following elliptic estimate 
(see Chapter 9 in \cite{GiTr} for instance)
\begin{align}
\label{elliptic00}
 \|v\!:\!{H^{k+2}(\Do_R)}\| \le
C(\|\Delta_x w\!:\!{H^k(\Do_R)}\| 
  +\|v\!:\!{L^2(\Do_R)}\|)
\end{align}
holds for $v \in H^{k+2}(\Do_R) \cap H^1_0(\Do_R)$ with a non-negative integer $k$.
On the other hand, we have
\begin{equation}\label{ha}
  \|v\!:\!{L^2(\Do_R)}\| \le C \|\nabla_x v\!:\!{L^2(\Do)}\|
\end{equation}
for $v \in H_0^1(\Do)$.
Indeed, one can show \eqref{ha} as follows.
We define a positive number $r(\omega)$ for each $\omega \in S^{1}$
so that $r(\omega)\,\omega \in \partial \Omega$, and
put $r_0=\mbox{dist}\,(0,\partial \Omega)$.
Then, for $v \in C^\infty_0({\Omega})$ we have
\begin{align}\nonumber
& |v(r\omega)|^2
  =\left| \int_{r(\omega)}^r (\omega\cdot\nabla v)(s\omega) ds
    \right|^2
\\ \notag
& \hspace{15mm}  \le \left(\int_{r(\omega)}^r \frac{ds}{s}\right)
 \left(\int_{r(\omega)}^r |\nabla v(s\omega)|^2 s ds\right)
\\ \nonumber
& \hspace{15mm} \le \frac{r}{r_0}
\int_{r(\omega)}^R |\nabla v(s\omega)|^2 s ds
\end{align}
for $r(\omega) <r<R$, because $r(\omega)\ge r_0$.
Multiplying it by $r$ and integrating the resulting inequality
over $\Omega_R$, we find \eqref{ha}.


Since $\varphi\in H^1_0(\Do)$ and $\text{supp}\,\chi \subset B_R$, we have $\varphi_1 \in H_0^1(\Do_R)$.
Therefore, the application of \eqref{elliptic00} 
in combination with (\ref{ha}) gives
\begin{equation}\label{ap10bis2}
 \| \varphi_1\!:\!{H^m(\Do)}\| \le
C(\|\Delta_x \varphi\!:\!{H^{m-2}(\Do)}\| 
  +\|\nabla_x \varphi\!:\!{L^2(\Do)}\|).
\end{equation}

Now our task is to show
\begin{equation}\label{ap10bis}
\sum_{|\alpha|=m} \| \pa_x^\alpha \varphi_2\!:\!{L^2(\Do)}\| \le
C(\|\Delta_x \varphi\!:\!{H^{m-2}(\Do)}\| 
  +\|\nabla_x \varphi\!:\!{L^2(\Do)}\|),
\end{equation}
because it implies \eqref{ap10} in view of \eqref{ap10bis2}.

Since $\|\pa^\alpha w\!:\!{L^2(\R^2)}\| \le
C\|\Delta_x w\!:\!{L^2(\R^2)}\|$ for $|\alpha|=2$ and 
$w \in H^2(\R^2)$, the left-hand side of \eqref{ap10bis} with $m=2$
is estimated by
\begin{equation}\nonumber
 C\|\Delta_x \varphi_2\!:\!{L^2(\Do)}\|
\le C(\|\Delta_x \varphi\!:\!{L^2(\Do)}\|
  +\|\nabla_x \varphi\!:\!{L^2(\Do)}\|
  +\|\varphi\!:\!{L^2(\Do_R)}\|).
\end{equation}
Hence, using \eqref{ha}, we obtain \eqref{ap10bis} for $m=2$.

For $k\ge 3$, similar argument to the above gives
\begin{equation} \notag
\sum_{|\alpha|=k}\norm{\pa_x^\alpha \varphi_2}{L^2(\Omega)}
\le C\bigl(\norm{\Delta_x \varphi}{H^{k-2}(\Do)}+\norm{\nabla_x\varphi}{H^{k-2}(\Omega)}\bigr),
\end{equation}
and the second term on the right-hand side 
is bounded by $C\bigl(\norm{\Delta_x \varphi}{H^{k-3}(\Do)}+\norm{\nabla_x\varphi}{L^2(\Do)}\bigr)$, if we know \eqref{ap10}
for $m=k-1$. Hence we inductively obtain \eqref{ap10bis} 
for $m\ge 2$.
\hfill$\qed$

\renewcommand{\theequation}{B.\arabic{equation}}
\setcounter{equation}{0}  
\renewcommand{\thelemma}{B.\arabic{theorem}}
\renewcommand{\thetheorem}{B.\arabic{theorem}}
\setcounter{theorem}{1}
\section*{Appendix B: Basic estimates for the Cauchy problem}

Here we prove the basic estimates used in the proof of Lemma \ref{freeD}.
In \cite{kov87}, \cite{ay}, and \cite{hk2}, some weighted $L^\infty$ estimates
for derivatives of solutions to the Cauchy problem are well examined.
However, we need their variants under different assumptions on the
right-hand member.
Although the proof of \eqref{der1} below can be done in a similar way as in the
previous works, we give its proof, because the case where $0<\nu<1$
has not been considered at all.

First of all, we introduce a couple of functions\,:
\begin{align} \notag
& K_1(\lam,\psi;r,t)=(2\pi)^{-1}
\{t^2-r^2-\lam^2+2r\lam\cos\psi\}^{-\frac12},  
\\
& K_2(\lam,\tau;r,t)=(2\pi)^{-1}
\{2r\lam\tau(1-\tau)(2-(1-\cos \vp)\tau)\}^{-\frac12},
\non
\\
\non
& \vp(\lam;r,t)=\arccos \bigg[\frac{r^2+\lam^2-t^2}{2r\lam}\bigg],
\\
& \Psi(\lam,\tau;r,t)=\mbox{arccos} [1-(1-\cos \vp(\lam;r,t))\tau],
\non
\\ \notag 
& K_3^{(\ell)} (\lam,\psi;r,t)=\frac{-(x_\ell-\lam\xi_\ell)}
 {2\pi (t^2-r^2-\lam^2+2r\lam \cos\psi)^{\frac32}}
\quad (\ell=1,2).
\end{align}
As for these functions, we shall use the following estimates.
For the proof of (\ref{kernel1}), (\ref{kernel1-}), and \eqref{kernel6}, see
for instance Proposition 5.3 in \cite{ay}.
Concerning  \eqref{kernel4-} and \eqref{kernel7}, see the proof of
(4.14) and (4.34) in \cite{hk2}, respectively.

\begin{lemma} \label{kova1}
We set $\lam_-=|t-s-r|$ and $\lam_+=t-s+r$.
If $0<s<t$ and $\lam_-<\lam<\lam_+$, then we have
\begin{align} \label{kernel1} 
& \int_{-\vp}^\vp K_1 (\lam,\psi;r,t-s) d\psi =2\int_0^1 K_2(\lam,\tau;r,t-s) d\tau 
\\ \notag
& \quad \le \frac{C}{(r\lam)^{\frac12}} 
  \log \bigg[2+\frac{r\lam}{(\lam-\lam_-)(\lam_++\lam)} H(t-s-r)\bigg],
\\  \label{kernel1-} 
& \int_0^1 |\pa_\lam K_2(\lam,\tau;r,t-s)| d\tau \le \frac{C}{(r\lam)^{\frac12}
 (\lam+s+r-t)}, 
\\ \label{kernel4-} 
& \int_0^1 |\pa_\lam \Psi \cdot K_2(\lam,\tau;r,t-s) | d\tau  
\\ \notag
& \quad \le \frac{C}{(r\lam)^\frac12} 
  \bigg( \frac1{\sqrt{(\lam_+-\lam)(\lam-\lam_-)}} 
  +\frac1{\sqrt{\lam^2-\lam_-^2}}  \bigg),
\end{align}
where $H(s)=1$ for $s>0$ and $H(s)=0$ otherwise.

On the other hand, if $0<s<t-r$ and $0 <\lam <\lam_-$, then we have 
\begin{align}\label{kernel6}
& \int_{-\pi}^\pi K_1(\lam,\psi;r,t-s) d\psi 
\\ \notag
& \quad \le \frac{C}{\sqrt{(\lam+\lam_-)(\lam_+-\lam)}} \log
\bigg[2+\frac{r\lam}{(\lam_--\lam)(\lam_++\lam)}\bigg],
\\
& \int_{-\pi}^{\pi}  |K_3^{(\ell)} (\lam,\psi;r,t-s)| d\psi
 \le \frac{C}{(\lam_--\lam)\sqrt{(\lam_-+\lam)(\lam_+-\lam)}}.
\label{kernel7}
\end{align}
\end{lemma}

Now we are in a position to state our basic estimates for solutions to
the Cauchy problem.

\begin{proposition}
Let $0<\nu<3/2$, $\mu\ge 0$, $\kappa \ge 1$, and $\eta>0$.
Then we have
\begin{align} \label{der1}
& | L_0[\pa_{t,x} g](x,t)| (1+|x|)^{\frac12}(1+|t-|x||)^{\nu}
\\ \notag
& \quad \le C \Psi_{1+\mu}(t+|x|)\,\Psi_{\kappa}(t+|x|)
  \| g(t) : M_1(z_{\nu+\mu,\kappa;0}) \|,
\end{align}
\begin{align} \label{der1Bis}
& | L_0[\pa_{t,x} g](x,t)| (1+|x|)^{\frac12}(1+|t-|x||)^{1-\eta}
\\ \notag
& \quad \le C \log(2+t+|x|)\, \| g(t) : M_1(z_{1,1;1}) \|
\end{align}
for $(x,t) \in \R^2 \times [0,T)$, where 
$C$ depends on $\nu$, $\mu$, $\kappa$, and $\eta$.
\end{proposition}

\begin{proof}
We prove only \eqref{der1}, because the other can be treated
analogously.
In addition, we evaluate only the spatial derivatives,
since the time derivative can be handled by using Proposition 5.3
in \cite{ay}.
Besides, since the case where $\mu>0$ is treated by modifying a little
the argument for handling the case $\mu=0$, we let $\mu=0$
in the following.

We set
\begin{eqnarray}
&& E_1=\{(y,s) \in \R^2 \times [0,t):\  |y|+s>t-r, \ |x-y|<t-s \},
\non
\\
&& E_2=\{(y,s) \in \R^2 \times [0,t):\ |y|+s<t-r \},
\non
\end{eqnarray}
so that 
$\overline{E_1} \cup \overline{E_2}=\{(y,s) \in \R^2 \times [0,t):\ |x-y|<t-s \}$.
According to this decomposition, we define
\begin{align}
P_j[g](x,t)=\frac1{2\pi}\iint_{E_j} \frac{g(y,s)}{\sqrt{(t-s)^2-|x-y|^2}}
  dy ds \quad (j=1,2).
\label{der2}
\end{align}
Then we have $L_0[\pa_\ell g](x,t)=P_1[\pa_\ell g](x,t)+P_2[\pa_\ell g](x,t)$
with $\ell=1,2$.

 Firstly we deal with $P_1[\pa_\ell g](x,t)$.
Following the computation made in the section 4 of \cite{hk}, we find that
\begin{equation}
|P_1[\pa_\ell g](x,t)| \le   \| g(t): M_1(z_{\nu,\kappa;0}) \|  \sum_{k=0}^5 I_{k},
\label{der5}
\end{equation}
where we have set
\begin{eqnarray}
I_{1}&=&
\iint_{D_1} \frac{\lam^{\frac12}}{z_{\nu,\kappa;0}
 (\lam,s)} d\lam ds \int_{-\vp}^{\vp}
 K_1(\lam,\psi;r,t-s) d\psi,
\nonumber \\
I_{2}&=&
\int_{D_2^{'}} 
\frac{\lam^{\frac12}}{z_{\nu,\kappa;0}
 (\lam,s)} d\sigma 
\int_{0}^{1} K_2(\lam,\tau;r,t-s) d\tau,
\nonumber \\
I_{3}&=&
\iint_{D_2} \frac1{\lam^{\frac12} z_{\nu,\kappa;0}
 (\lam,s)} d\lam ds \int_{0}^{1}
 K_2(\lam,\tau;r,t-s) d\tau, 
\nonumber \\
I_{4}&=&
\iint_{D_2} 
\frac{\lam^{\frac12}}{z_{\nu,\kappa;0}
 (\lam,s)} d\lam ds \int_{0}^{1}
 |\pa_\lam K_2(\lam,\tau;r,t-s)| d\tau, 
\nonumber \\
I_{5}&=&
\iint_{D_2} \frac{\lam^{\frac12}}{z_{\nu,\kappa;0}
 (\lam,s)} d\lam ds \int_{0}^{1}
 |(\pa_\lam \Psi \cdot K_2)(\lam,\tau;r,t-s)| d\tau
\nonumber
\end{eqnarray}
and
\begin{align}
& D_1 = \{(\lam,s)\in (0,\infty) \times (0,t):\
\lam_-<\lam \le \lam_- +\delta \ \mbox{or} \
\lam_+ -\delta \le \lam <\lam_+ \},
\nonumber \\
& D_2 =\{(\lam,s)\in (0,\infty) \times (0,t):\
\lam_- +\delta \le \lam \le \lam_+ -\delta \},
\nonumber
\\
& D_2^{'}=\{(\lam,s) \in (0,\infty) \times (0,t):\
\lam=\lam_- +\delta \ \mbox{or} \ \lam=\lam_+ -\delta \}
\non
\end{align}
with $\delta=\min\{r,1/2\}$ and $\lam_-=|t-s-r|$, $\lam_+=t-s+r$.

Now we are going to show that
\begin{equation}
I_k \le C (1+r)^{-\frac12} (1+|t-r|)^{-\nu} \log(2+t+r)\,\Psi_{\kappa}(t+r)
\label{der11}
\end{equation}
holds for $k=1,\dots,5$.
First we evaluate $I_1$.
Notice that when $0<s<t-r$ and $\lam>\lam_+-\delta$, we have
$\lam-\lam_- >r$, so that
$$
\log \bigg[2+\frac{r\lam}{(\lam-\lam_-)(\lam_++\lam)} \bigg]
\le \log 3.
$$
For $0<s<t-r$ and $\lam>\lam_-$, we get
\begin{equation}
\log \bigg[2+\frac{r\lam}{(\lam-\lam_-)(\lam_++\lam)} \bigg]
\le \log \bigg[2+\frac{\lam}{\lam-\lam_-} \bigg].
\label{der8}
\end{equation}
Moreover, we note that $z_{\nu,\kappa;0}(\lam,s)$
is equivalent to $z_{\nu,\kappa;0}(\lam_+,s)$~(resp.
$z_{\nu,\kappa;0}(\lam_-,s)$) for
$\lam_+-\delta<\lam<\lam_+$~(resp. $\lam_- < \lam < \lam_-+\delta$).
Hence by (\ref{kernel1}), we get
\begin{equation}
I_1 \le C r^{-\frac12} [A_{1,0}+A_{2,0}+A_{3,0}],
\label{der6}
\end{equation}
where we have set
\begin{eqnarray}
A_{1,0}&=&\int_0^t \int_{\lam_+-\delta}^{\lam_+}
\frac{1}{z_{\nu,\kappa;0}(\lam_+,s)} d\lam ds,
\non
\\
A_{2,0}&=&\int_0^{(t-r)_+} \int_{\lam_-}^{\lam_-+\delta}
  \frac{1}{z_{\nu,\kappa;0}(\lam_-,s)}
   \log \bigg[2+\frac{\lam}{\lam-\lam_-} \bigg] d\lam ds,
\non
\\
A_{3,0}&=&\int_{(t-r)_+}^t \int_{\lam_-}^{\lam_-+\delta}
 \frac{1}{z_{\nu,\kappa;0}(\lam_-,s)}  d\lam ds.
\non
\end{eqnarray}

\par
It is easy to see that
\begin{eqnarray}
A_{1,0}&\le & \frac{\delta}{(1+t+r)^{\nu}} 
 \int_0^t \frac{1}{(1+\lambda_+)^{\kappa}} ds
\label{der7}
\\
&\le & C \delta(1+t+r)^{-\nu} \Psi_{\kappa}(t+r).
\non
\end{eqnarray}

To evaluate $A_{2,0}$, observe that
\begin{align}
\int_{\lam_-}^{\lam_-+\delta}
   \log \bigg[2+\frac{\lam}{\lam-\lam_-} \bigg] d\lam 
\le C\delta^{1/2} \log(3+\lambda_-).
\end{align}
Indeed, the left-hand side is equal to
$$
\delta \log \bigg[3+\frac{\lam_-}{\delta} \bigg] 
 +\int_{\lam_-}^{\lam_-+\delta}
   \frac{\lam_-}{3(\lam-\lam_-)+\lam_-} d\lam, 
$$
which is bounded by the right-hand side, if we use
$0\le \delta\le 1/2$, and 
an inequality $|x^{1/2} \log x| \le 2e^{-1}$ for $0<x<1$.

Since $s+\lambda_- \ge |t-r|$, we get
\begin{eqnarray}
A_{2,0}&\le & \frac{C\delta^{1/2}}{(1+|t-r|)^{\nu}} 
 \int_{0}^{(t-r)_+} \frac{\log(3+\lambda_-)}{(1+\lambda_-)^{\kappa}} ds
\label{der9}
\\
&\le & C\delta^{1/2}(1+|t-r|)^{-\nu} (\Psi_{\kappa}(t-r))^2.
\non
\end{eqnarray}
We easily have
\begin{eqnarray}
A_{3,0} \le C\delta(1+|t-r|)^{-\nu}\Psi_{\kappa}(t+r).
\label{der10}
\end{eqnarray}
Summing up (\ref{der7}), (\ref{der9}) and (\ref{der10}), we see
from \eqref{der6} that \eqref{der11} holds for $k=1$.

\par
 In the following, we assume $r \ge 1/2$ so that $\delta=1/2$,
because $D_2$ is the empty set when $0<r<1/2$.
Since $\lam=\lam_-+(1/2)$ or $\lam=\lam_+-(1/2)$
for $(\lam,s) \in D_2^{'}$, we get \eqref{der11} for $k=2$ 
similarly to the previous argument.

\par
 Next we evaluate $I_3$.
Note that $\lam \ge 1/2$ if $(\lam,s) \in D_2$ and that
$$
\log \bigg[2+\frac{r\lam}{(\lam-\lam_-)(\lam_++\lam)} \bigg]
\le C \log (2+\lam)
$$
for $0<s<t$ and $\lam \ge \lam_-+(1/2)$.
Therefore we get from (\ref{kernel1})
\begin{eqnarray} \notag
r^{\frac12} I_3 \le C \iint_{D_2} \frac{ \log (2+\lam)\,d\lam ds}
  {(1+\lam) z_{\nu,\kappa; 0}(\lam,s)}
=C \iint_{D_2} \frac{ \log (2+\lam)\,d\lam ds}{(1+s+\lam)^{\nu}
   (1+\lam)^{1+\kappa}}.
\end{eqnarray}
Since $s+\lam \ge |t-r|$ for $(\lam,s) \in D_2$, the right-hand side is
bounded by
\begin{eqnarray}
\frac{C}{(1+|t-r|)^{\nu}} \iint_{D_2} \frac{ \log (2+\lam)\,d\lam ds}
 {(1+\lam)^{1+\kappa}}
\le  \frac{C(\Psi_\kappa(t,r))^2}{(1+|t-r|)^{\nu}}.
\non
\end{eqnarray}
Therefore, (\ref{der11}) holds for $k=3$.

Next we evaluate $I_4$.
Since $\lam +s+r-t\ge 1/2$ for $\lam\ge\lam_- +(1/2)$,
we get from (\ref{kernel1-}) 
\begin{eqnarray}
r^{\frac12}I_4 &\le& C
 \iint_{D_2} \frac{ d\lam ds}{ z_{\nu,\kappa; 0}(\lam,s)
 (\lam+s+r-t+1)}
\non
\\ \notag
&\le& C \int_{|t-r|}^{t+r} \frac{d\alpha}{(\alpha-t+r+1)
      (1+\alpha)^{\nu}}
  \int_{-\alpha}^{\alpha}
      \frac{d\beta}{(1+\alpha+\beta)^{\kappa}}
\\
\non
&\le&  \frac{C}{(1+|t-r|)^{\nu}}\int_{|t-r|}^{t+r}
 \frac{\Psi_\kappa(\alpha)}{(\alpha-t+r+1)}
 d\alpha
\\ \notag
&\le&  C(1+|t-r|)^{-\nu} \Psi_\kappa(t+r)\,{\rm log}(2+r),
\end{eqnarray}
where we have changed the variables by
\begin{align} \label{b16}
\alpha=\lambda+s, \quad \beta=\lambda-s.
\end{align}

Next we evaluate $I_5$.
It follows from (\ref{kernel4-}) that
\begin{eqnarray}
\non
r^{\frac12}I_5 \le C
(A_{5,0}+B_{5,0}+C_{5,0}),
\end{eqnarray}
where we have set 
\begin{eqnarray}
\non
A_{5,0}&=& \iint_{D_2} \frac{ d\lam ds}{
 z_{\nu,\kappa; 0}(\lam,s)
\sqrt{t-s+r-\lam+1}\sqrt{\lam-t+s+r+1}},
\\
\non
B_{5,0}&=& \iint_{D_2} \frac{ d\lam ds}{
  z_{\nu,\kappa; 0}(\lam,s)
\sqrt{t-s+r-\lam+1}\sqrt{\lam+t-s-r+1}},
\\
\non
C_{5,0}&=& \iint_{D_2} \frac{ d\lam ds}{
 z_{\nu,\kappa; 0}(\lam,s)
\sqrt{\lam-t+s+r+1}\sqrt{\lam+t-s-r+1}}.
\end{eqnarray}
Changing the variables by (\ref{b16}), we have
\begin{eqnarray}
A_{5,0} &\le& C \int_{|t-r|}^{t+r}
 \frac{d\alpha}{(1+\alpha)^{\nu}
\sqrt{t+r-\alpha}\sqrt{\alpha-t+r}} 
   \int_{-\alpha}^\alpha
  \frac{d\beta}{(1+\alpha+\beta)^{\kappa}}
\non
\\ \notag
&\le & C (1+|t-r|)^{-\nu}
\int_{t-r}^{t+r}
 \frac{\Psi_\kappa(\alpha)\,d\alpha}{\sqrt{t+r-\alpha}\sqrt{\alpha-t+r}} 
\\
&=& C \pi (1+|t-r|)^{-\nu} \Psi_\kappa(t+r).
\non
\end{eqnarray}
Moreover, we see that $B_{5,0}$ is bounded by
\begin{align}
  C \int_{|t-r|}^{t+r}
 \frac{d\alpha}{(1+\alpha)^{\nu}
\sqrt{t+r-\alpha+1}} 
   \int_{r-t}^\alpha
  \frac{d\beta}{(1+\alpha+\beta)^{\kappa}
\sqrt{t-r+\beta+1} }.
\non
\end{align}
By integrating by parts in the $\beta$-integral, it is estimated by
$$
2(1+2\alpha)^{(1/2)-\kappa}+2\kappa \int_{r-t}^\alpha
 (1+\alpha+\beta)^{-(1/2)-\kappa} d\beta
\le C(1+\alpha+r-t)^{(1/2)-\kappa},
$$
if $\alpha \ge |r-t|$ and $\kappa>1/2$.
Therefore we get
\begin{eqnarray}
(1+|t-r|)^{\nu} B_{5,0} \le C \int_{|t-r|}^{t+r}
 \frac{d\alpha}{\sqrt{t+r-\alpha+1}\, (1+\alpha+r-t)^{\kappa-(1/2)}}
\le C,
\non
\end{eqnarray}
for $\kappa \ge 1$.
Similarly, one can show
\begin{eqnarray}
(1+|t-r|)^{\nu} C_{5,0} \le C\,\Phi_\kappa(t+r).
\non
\end{eqnarray}
Thus we obtain (\ref{der11}) for all $k=1,\dots,5$ in conclusion.

\par
 Secondly we deal with $P_2[\pa_\ell g](x,t)$. 
First, suppose $0 \le t-r \le 2$.
Switching to the polar coordinates as
\begin{equation}
x=(r\cos \theta, r\sin \theta), \quad
y=\lam\xi=(\lam\cos(\theta+\psi), \lam\sin(\theta+\psi))
\label{der3}
\end{equation}
in \eqref{der2} with $j=2$, we get
\begin{equation}
P_2[\pa_\ell g](x,t)=\int_0^{t-r}\!\! \int_{0}^{t-s-r}\!\! \int_{-\pi}^{\pi} 
\lam \pa_\ell g(\lam\xi,s)
  K_1(\lam,\psi;r,t-s) d\psi d\lam ds.
\non 
\end{equation}
For $0<s<t-r$, $0< \lam <\lam_-$, we have
$0<{\lam_- -\lam} \le 2$, so that
\eqref{kernel6} yields
\begin{align}
& \int_{-\pi}^\pi K_1(\lam,\psi;r,t-s) d\psi 
  \le \frac{C}
{\sqrt{\lam} \sqrt{\lam_+ -\lam}} 
\log \left[ 2+\frac{\lambda}{\lam_--\lam} \right]
\label{kernel5}
\\ \notag
& \hspace{30mm}  \le \frac{C}
{\sqrt{\lam} \sqrt{r+1} \sqrt{\lam_- -\lam}} 
\log \left[ 2+\frac{\lam}{\lam_--\lam} \right].
\end{align}
Since $\lam <\lam_- \le 2$, we get
\begin{align}\notag
& \sqrt{r+1}\,|P_2[\pa_\ell g](x,t)| 
\\ \notag
& \quad  \le C \| g(t) : M_1({z}_{\nu,\kappa;0}) \|
 \int_0^{t-r}\!\! \int_{0}^{\lambda_-}
 \frac{1}{\sqrt{\lam_--\lam}}
\log \left[ 2+\frac{2}{\lam_--\lam} \right] d\lam ds.
\end{align}
The last integral is bounded, because $0 \le t-r \le 2$. 
Hence we obtain
\begin{align}\label{p20}
(1+r)^{1/2}  (1+|t-r|)^{\nu} |P_2[\pa_\ell g](x,t)| 
 \le C \| g(t) : M_1({z}_{\nu,\kappa;0}) \|.
\end{align}

In the following, suppose $t-r \ge 2$, so that $t-r-1 \ge (t-r)/2$.
We decompose $P_2[\pa_\ell g](x,t)$ as
\begin{align*}
P_2[\pa_\ell g](x,t) = & 
 \frac1{2\pi}\iint_{E_2} \frac{H(|y|>t-s-r-1)\,\pa_\ell g(y,s)}{\sqrt{(t-s)^2-|x-y|^2}}
  dy ds 
\\
& \  +\frac1{2\pi}\iint_{E_2} \frac{H(|y|<t-s-r-1)\,\pa_\ell g(y,s)}{\sqrt{(t-s)^2-|x-y|^2}}
  dy ds 
\\
\equiv & \, Q_1(x,t)+Q_2(x,t).
\end{align*}
Since $0<\lam_- -\lam <1$ for $t-s-r-1 < \lam <{t-s-r}$, 
one can proceed as in the previous case and get
\begin{align} \notag 
& \sqrt{r+1}\,|Q_1(x,t)| \le C \| g(t) : M_1({z}_{\nu,\kappa;0}) \|
\\ \notag 
& \quad \times \int_0^{t-r}\!\! \int_{(t-s-r-1)_+}^{t-s-r}
 \frac{\log(2+\lam)-\log(\lam_--\lam)}{z_{\nu,\kappa;0}(s,\lambda) \sqrt{\lam_--\lam}}
 d\lam ds.
\end{align}
Changing the variables by (\ref{b16}), the last integral is estimated by
\begin{align} \notag
& C \int_{t-r-1}^{t-r}
 \frac{d\alpha}{(1+\alpha)^{\nu}(t-r-\alpha)^{\frac12}}
  \int_{-\alpha}^{\alpha}
       \frac{\log(4+\alpha+\beta)-\log(t-r-\alpha)}{(1+\alpha+\beta)^{\kappa}}
   d\beta
\\ \notag
&  \ \le C(1+|t-r|)^{-\nu} (\Psi_\kappa(t-r))^2.
\end{align}
Thus we get
\begin{align} \label{p2}
& \sqrt{r+1}\,(1+|t-r|)^{\nu} |Q_1(x,t)| 
\\ \notag
& \quad \le C \| g(t) : M_1({z}_{\nu,\kappa;0}) \| (\Psi_\kappa(t-r))^2.
\end{align}

 Finally, we deal with $Q_2(x,t)$.
Making the integration by parts in $y$ and
switching to the polar coordinates as (\ref{der3}), we get
\begin{align}
 Q_2(x,t) 
  &=\int_0^{t-r-1}\!\! \int_0^{t-s-r-1}\!\! \int_{-\pi}^{\pi}  \lam g(\lam\xi,s)
  K_3^{(\ell)} (\lam,\psi;x,t-s) d\psi d\lam ds
\label{p3}
\\
& \quad +\int_0^{t-r-1}\!\! \int_{-\pi}^{\pi} \lam\xi_\ell g(\lam\xi,s)
    K_1(\lam,\psi;r,t-s) \bigg|_{\lam=t-s-r-1} d\psi ds.
\non
\end{align}
We see from (\ref{kernel5}) that the second term in the right-hand side of (\ref{p3})
is evaluated by $C \| g(t) : M_0({z}_{\nu,\kappa;0}) \|/\sqrt{r+1}$ times
\begin{eqnarray}\notag
&& \int_0^{t-r-1}
 \frac1{z_{\nu,\kappa;0}(\lam,s) \sqrt{\lam_- -\lam}}  
\log \left[ 2+\frac{\lambda}{\lam_--\lam} \right] \bigg|_{\lam=t-s-r-1} ds
\\
&\le&
 C (1+|t-r|)^{-\nu}(\Psi_\kappa(t-r))^2.
\non
\end{eqnarray}
Noting $\lam_+-\lam \ge 2r+1$ for $\lam<t-s-r-1$, 
we see from \eqref{kernel7} that the first term in the right-hand side 
of (\ref{p3}) is estimated by
$C \| g(t) : M_0({z}_{\nu,\kappa;0}) \|/\sqrt{r+1}$ times
\begin{align}
& \int_0^{t-r-1}\!\! \int_{(\frac{t-r}2-s)_+}^{t-s-r-1}
 \frac{ d\lam ds}{z_{\nu,\kappa;0}(\lam,s)
 (\lam_- -\lam)}
+\int_0^{\frac{t-r}2}\!\! \int_0^{\frac{t-r}2-s}
 \frac{ \sqrt{\lam} d\lam ds}{z_{\nu,\kappa;0}(\lam,s)
(\lam_- -\lam)^{\frac32}}
\non
\\
& \quad \le \frac{C}{(1+|t-r|)^{\nu}}
 \int_{\frac{t-r}2}^{t-r-1} \frac{d\alpha}{t-r-\alpha}
  \int_{-\alpha}^{\alpha} \frac{d\beta}
 {(1+\alpha+\beta)^{\kappa}}
\non
\\
& \quad \quad +\frac{C}{(1+|t-r|)^{\frac32}}
 \int_0^{\frac{t-r}2} \frac{d\alpha}{(1+\alpha)^{\nu-\frac12}}
  \int_{-\alpha}^{\alpha} \frac{d\beta}  {(1+\alpha+\beta)^{\kappa}}
\non
\\
& \quad \le C(1+|t-r|)^{-\nu} \log (2+t-r)\,\Psi_\kappa(t-r),
\non
\end{align}
because we have assumed $\nu<3/2$.
Thus we get
\begin{align} \label{p31}
& \sqrt{r+1}\,(1+|t-r|)^{\nu} |Q_2(x,t)| 
\\ \notag
& \quad \le C \| g(t) : M_0({z}_{\nu,\kappa;0}) \| 
    \log (2+t-r)\,\Psi_\kappa(t-r).
\end{align}
Now, (\ref{der1}) follows from (\ref{der5}), (\ref{der11}), \eqref{p20}, \eqref{p2}
and (\ref{p31}).
This completes the proof of Proposition B.2.
\end{proof}

\begin{center}
{\bf Acknowledgments}
\end{center}
The author is grateful to Professor Sandra Lucente for valuable discussion.
He is partially supported by Grant-in-Aid for Science Research No. 20224013, JSPS.




\begin{thebibliography}{99}

\bibitem{ay}R. Agemi and K. Yokoyama,
{\it The null condition and global existence of
solutions to systems of wave equations with different speeds},
{\sl in }
\lq\lq Advances in nonlinear partial differential equations and
stochastics" (S. Kawashima and T. Yanagisawa ed.),
Series on Adv. in Math. for Appl. Sci., Vol. 48,
43--86, World Scientific, Singapore, 1998.



\bibitem{DiF03}M. Di Flaviano,
{\it
Lower bounds of the life span of classical solutions to a system
of semilinear wave equations in two space dimensions},
 J. Math. Anal. Appl.
{\bf 281} (2003), 22--45.

\bibitem{GiTr}D.~Gilbarg and N.~S.~Trudinger,
\lq\lq Elliptic partial differential equations of second order",
Second edition, Springer-Verlag, Berlin, 1983.

\bibitem{God89}P. Godin,
{\it Long time behavior of solutions to some nonlinear invariant
mixed problems}, Comm. Partial Differential Equations {\bf 14}
(1989), 299--374.

\bibitem{God95}P. Godin,
{\it Global existence of solutions to some exterior radial
quasilinear Cauchy-Dirichlet problems}, 
Amer. J. Math. {\bf 117}
(1995), 1475--1505.

\bibitem{Ha95}N. Hayashi,
{\it
Global existence of small solutions to quadratic nonlinear wave
equations in an exterior domain},
 J. Funct. Anal.
{\bf 131} (1995), 302--344.

\bibitem{hk}A.~Hoshiga and H.~Kubo,
{\it
Global small amplitude solutions of nonlinear hyperbolic 
systems with a critical exponent 
under the null condition},
SIAM J.~Math.~Anal. {\bf 31}
(2000),~486--513. 

\bibitem{hk2}A.~Hoshiga and H.~Kubo,
{\it
Global solvability for systems of nonlinear wave equations
with multiple speeds in two space dimensions},
 Diff.~Integral Eqs.
{\bf 17} (2004), 593--622.




 
\bibitem{KatKub}S.~Katayama and H.~Kubo,
{\it 
An alternative proof of global existence for nonlinear wave equations 
in an exterior domain},
J. Math. Soc. Japan {\bf 60} (2008), 1135--1170. 

\bibitem{KeSmiSo02G}M.~Keel, H.~Smith and C.~D.~Sogge,
{\it Global existence for a quasilinear wave equation outside of
star-shaped domains}, J.~Funct.~Anal. {\bf 189} (2002), 155--226.

\bibitem{KeSmiSo04}M. Keel, H. Smith and C. D. Sogge,
{\it Almost global existence for quasilinear wave equations in
three space dimensions}, J. Amer. Math. Soc. {\bf 17} (2004), 109--153.

\bibitem{kl0}S. Klainerman,
{\it Uniform decay estimates and the Lorentz invariance
of the classical wave equation},
 Comm. Pure Appl. Math.
{\bf 38} (1985), 321--332.

\bibitem{kov87}M. Kovalyov,
{\it
Long-time behaviour of solutions of a system on
non-linear wave equations},
Comm. in P.D.E.
{\bf 12} (1987), 471--501.

\bibitem{Kub94}H.~Kubo,
{\it
Blow-up for semilinear wave equations with initial data
of slow decay in low space dimensions},
Differential and Integral Equations,
{\bf 7} (1994), 315--321.

\bibitem{Kub06}H.~Kubo,
{\it
Uniform decay estimates for the wave equation in an exterior domain},
~{\rm in} \lq\lq Asymptotic analysis and singularities", 31--54,
Advanced Studies in Pure Mathematics 47-1, Math.~Soc. of Japan, 2007.

\bibitem{kubota}K. Kubota,
{\it
Existence of a global solutions to a semi-linear wave equation
with initial data of non-compact support in low space dimensions},
 Hokkaido Math. J. 
{\bf 22} (1993), 123--180.



\bibitem{Met04}J.~Metcalfe,
{\it
Global existence for semilinear wave equations exterior to
nontrapping obstacles},
 Houston J.~Math.
{\bf 30} (2004), 259--281.

\bibitem{MetNaSo05b}J.~Metcalfe, M.~Nakamura and C.~D.~Sogge,
{\it
Global existence of quasilinear, nonrelativistic wave equations
satisfying the null condition},
 Japan.~J.~Math. (N.S.) 
{\bf 31} (2005), 391--472.

\bibitem{MetSo05}J.~Metcalfe and C.~D.~Sogge,
{\it
Hyperbolic trapped rays and global existence of quasilinear wave equations},
 Invent.~Math.
{\bf 159} (2005), 75--117.


\bibitem{Mor75}C. S.~Morawetz,
{\it
Decay for solutions of the exterior problem for the wave equation},
 Comm. Pure Appl. Math.
{\bf 28} (1975), 229--264.



\bibitem{ShiTsu86}Y.~Shibata and Y.~Tsutsumi,
{\it
On a global existence theorem of small amplitude solutions
for nonlinear wave equations in an exterior domain},
 Math.~Z.
{\bf 191} (1986), 165--199.
 
\bibitem{Vai75} B. R. Vainberg,
{\it
The short-wave asymptotic behavior of the solutions of stationary problems, 
and the asymptotic behavior as $t\rightarrow \infty $ of the solutions of 
nonstationary problems},
(Russian) Uspehi Mat. Nauk 
{\bf 30} (1975), 3--55. 



\end{thebibliography}
\end{document}